\documentclass[sn-mathphys,Numbered]{sn-jnl}

\usepackage{graphicx}%
\usepackage{subcaption}
\usepackage{amsmath}%
\usepackage{amsthm}%
\usepackage{mathrsfs}
\usepackage[title]{appendix}%
\usepackage{textcomp}%
\usepackage{manyfoot}%
\usepackage{algorithm}%
\usepackage{algorithmicx}%
\usepackage{algpseudocode}%
\usepackage{amssymb}
\usepackage{amsfonts}
\usepackage{mathtools}
\usepackage{natbib}
\setcitestyle{square, comma, numbers,sort&compress, super}

\usepackage{url}
\usepackage{stmaryrd}
\usepackage[table]{xcolor}
\usepackage{xcolor}
\usepackage{xspace}
\usepackage{thmtools} 
\usepackage{thm-restate} 
\usepackage{tikz}
\usepackage{caption}
\usepackage{subcaption}
\usepackage{makecell}
\usepackage{multirow}
\usepackage{booktabs}
\usepackage[normalem]{ulem}
\usepackage[usestackEOL]{stackengine}
\definecolor{codegreen}{rgb}{0,0.6,0}
\definecolor{codegray}{rgb}{0.5,0.5,0.5}
\definecolor{codepurple}{rgb}{0.58,0,0.82}
\definecolor{backcolour}{rgb}{0.95,0.95,0.92}
\usepackage{listings}
\lstdefinestyle{mystyle}{
	backgroundcolor=\color{backcolour},
	commentstyle=\color{codegreen},
	keywordstyle=\color{magenta},
	numberstyle=\tiny\color{codegray},
	stringstyle=\color{codepurple},
	basicstyle=\ttfamily\footnotesize,
	breakatwhitespace=false,
	breaklines=true,
	captionpos=b,
	keepspaces=true,
	numbers=left,
	numbersep=4pt,
	showspaces=false,
	showstringspaces=false,
	showtabs=false,
	tabsize=2,
	escapeinside={*@}{@*}
}
\usepackage[misc]{ifsym}
\usepackage[frozencache=true,cachedir=_minted-cache]{minted}
\usepackage{pifont}
\usepackage[capitalize]{cleveref}
\usepackage[shortlabels]{enumitem}

\newif\iffull\fullfalse

\newcommand{\post}{\mathsf{postE}\xspace}
\newcommand{\pre}{\mathsf{preE}\xspace}

\newcommand{\getfeatures}{\textsf{getFeatures}\xspace}
\newcommand{\getstates}{\textsf{sampleStates}\xspace}
\newcommand{\sample}{\textsf{sampleTraces}\xspace}
\newcommand{\learninv}{\textsf{learnInv}\xspace}
\newcommand{\extractinv}{\textsf{extractInv}\xspace}
\newcommand{\verifyinv}{\textsf{verifyInv}\xspace}

\newcommand{\pprog}{\prog\xspace}
\newcommand{\pexp}{\mathit{pexp}\xspace}
\newcommand{\feat}{\mathit{feat}\xspace}
\newcommand{\nruns}{N_{\mathit{runs}}\xspace}
\newcommand{\nstates}{N_{\mathit{states}}\xspace}
\newcommand{\states}{\mathit{states}\xspace}

\newcommand{\nndt}{\mathsf{nndt}\xspace}
\newcommand{\labeling}{\mathsf{nn}_{\text{label}}\xspace}
\newcommand{\regress}{\mathsf{nn}_{\text{regress}}\xspace}
\newcommand{\classify}{\mathsf{nn}_{\text{classify}}\xspace}
\newcommand{\nnmt}{\mathsf{nnmt}\xspace}

\newcommand{\expr}{e}
\newcommand{\dist}{\mu}
\newcommand{\Dist}[1]{\mathbf{Dist}(#1)}
\newcommand{\States}{\Sigma}
\newcommand{\Expt}{\mathcal{E}}
\newcommand{\Tree}{T}
\newcommand{\LANG}{$\mathbf{pWhile}$\xspace}
\newcommand{\prog}{\bench{geo}\xspace}
\newcommand{\weakp}{\mathsf{wpe}}

\newcommand{\lfp}{\mathsf{lfp}}
\newcommand{\unit}{\mathsf{unit}\xspace}
\newcommand{\bind}{\mathsf{bind}\xspace}
\newcommand{\defeq}{\coloneqq}
\newcommand{\defdefeq}{\coloneqq}

\newcommand{\Skip}{\mathbf{skip}}

\newcommand{\Seq}[2]{{#1} \mathrel{;} {#2}}
\newcommand{\Assn}[2]{\ensuremath{{#1} \leftarrow {#2}}}
\newcommand{\Rand}[2]{{#1} \stackrel{\raisebox{-.25ex}[.25ex]%
		{\tiny $\mathdollar$}}{\raisebox{-.2ex}[.2ex]{$\leftarrow$}} {#2}}
\newcommand{\while}[2]{\mathbf{while}\ #1\ \mathbf{:}\ #2}

\newcommand{\Cond}[3]{\mathbf{if}\ #1\ \mathbf{then}\ #2\ \mathbf{else}\ #3}

\newcommand{\bench}[1]{\textsf{#1}}
\newcommand{\kwtrue}{\mathit{tt}}
\newcommand{\kwfalse}{\mathit{ff}}
\newcommand{\Leaf}{\mathsf{Leaf}}
\newcommand{\Node}{\mathsf{Node}}

\newcommand{\geozero}{\bench{Geo0}\xspace}
\newcommand{\geoone}{\bench{Geo1}\xspace}
\newcommand{\geotwo}{\bench{Geo2}\xspace}
\newcommand{\fair}{\bench{Fair}\xspace}
\newcommand{\mart}{\bench{Mart}\xspace}
\newcommand{\gambler}{\bench{Gambler}\xspace}
\newcommand{\geoar}{\bench{GeoAr}\xspace}
\newcommand{\binzero}{\bench{Bin0}\xspace}
\newcommand{\binone}{\bench{Bin1}\xspace}
\newcommand{\bintwo}{\bench{Bin2}\xspace}

\newcommand{\linexp}{\bench{LinExp}\xspace}

\newcommand{\progsum}{\bench{Sum0}\xspace}
\newcommand{\deprv}{\bench{DepRV}\xspace}

\newcommand{\biasdir}{\bench{BiasDir}\xspace}
\newcommand{\prinsys}{\bench{Prinsys}\xspace}
\newcommand{\duel}{\bench{Duel}\xspace}

\newcommand{\detm}{\bench{Detm}\xspace}
\newcommand{\revbin}{\bench{RevBin}\xspace}
\newcommand{\denot}[1]{\llbracket #1 \rrbracket}
\newcommand{\tool}{\mbox{\textsc{Exist}}\xspace}

\newcommand{\fix}[1]{\textcolor{black}{#1}}%

\theoremstyle{thmstyleone}%
\newtheorem{theorem}{Theorem}
\newtheorem{proposition}[theorem]{Proposition}%

\theoremstyle{thmstyletwo}%

\theoremstyle{thmstylethree}%
\newtheorem{definition}{Definition}%

\begin{document}

\title[Article Title]{Data-Driven Invariant Learning for Probabilistic Programs}


\author[1]{\fnm{Jialu} \sur{Bao}}\email{jb965@cornell.edu}

\author[2]{\fnm{Nitesh} \sur{Trivedi}}\email{nitesht@iitk.ac.in}

\author[3]{\fnm{Drashti} \sur{Pathak}}\email{dpathak1997@gmail.com}

\author[1]{\fnm{Justin} \sur{Hsu}}\email{justin@cs.cornell.edu}

\author[2]{\fnm{Subhajit} \sur{Roy}}\email{subhajit@iitk.ac.in}

\affil[1]{\orgname{Cornell University}, \orgaddress{\city{Ithaca}, \state{NY}, \country{USA}}}

\affil[2]{\orgname{Indian Institute of Technology (IIT) Kanpur}, \orgaddress{\country{India}}}

\affil[3]{\orgname{Amazon}, \orgaddress{\state{Bangalore}, \country{India}}}


\abstract{Morgan and McIver's \emph{weakest pre-expectation} framework is one of the most well-established methods for deductive verification of probabilistic
	programs. Roughly, the idea is to generalize binary state assertions to
	real-valued \emph{expectations}, which can measure expected values of
	probabilistic program quantities. While loop-free programs can be analyzed by
	mechanically transforming expectations, verifying loops usually requires
	finding an \emph{invariant expectation}, a difficult task.

	We propose a new view of invariant expectation synthesis as a
	\emph{regression} problem: given an input state, predict the \emph{average}
	value of the post-expectation in the output distribution. Guided by this
	perspective, we develop the first \emph{data-driven} invariant synthesis
	method for probabilistic programs. Unlike prior work on probabilistic
	invariant inference, our approach can learn piecewise continuous invariants
	without relying on template expectations. We also develop a data-driven
	approach to learn \emph{sub-invariants} from data, which can be used to upper-
	or lower-bound expected values. We implement our approaches and demonstrate
	their effectiveness on a variety of benchmarks from the probabilistic
	programming literature.}




\maketitle

\section{Introduction} \label{sec:intro}

\emph{Probabilistic programs} are imperative programs augmented with a \textit{sampling} command that allows programs to draw from probability distributions.
Probabilistic programs provide a natural way to express randomized computations. While
the mathematical semantics of such programs is fairly well-understood
\citep{Kozen81}, verification methods remain an active area of research.
Existing automated techniques are either limited to specific properties
(e.g.,~\cite{SHA18,AH17,DBLP:conf/oopsla/CarbinMR13,Kolahal}), or target simpler
computational
models~\cite{DBLP:conf/icalp/BaierCHKR97,kwiatkowska2011prism,DBLP:journals/corr/DehnertJK017}.

\paragraph*{Reasoning about Expectations.}
One of the earliest methods for reasoning about probabilistic programs is
through \emph{expectations}. Originally proposed by Kozen~\cite{Kozen:1985},
expectations generalize standard, binary assertions to quantitative, real-valued
functions on program states. Morgan and McIver further developed this idea into
a powerful framework for reasoning about probabilistic imperative programs,
called the \emph{weakest pre-expectation calculus}~\cite{Morgan:1996,McIver:2005}.

Concretely, Morgan and McIver defined an operator called the
\emph{weakest pre-expectation} ($\weakp$), which takes an expectation $E$ and a
program $P$ and produces an expectation $E'$ such that $E'(\sigma)$ is the expected value
of $E$ in the output distribution $\denot{P}_\sigma$. In this way, the $\weakp$
operator can be viewed as a generalization of Dijkstra's weakest
pre-conditions calculus~\cite{dijkstra-wp} to probabilistic programs. For
verification purposes, the $\weakp$ operator has two key strengths. First, it
enables reasoning about probabilities and expected values. Second,
when $P$ is a loop-free program, it is possible
to transform $\weakp(P, E)$ into a form that does not mention the program $P$
via simple, mechanical manipulations, essentially analyzing the effect of the
program on the expectation through syntactically transforming $E$.

However, there is a caveat: the $\weakp$ of a loop is defined as a least fixed
point, and it is generally difficult to simplify this quantity into a more
tractable form. Fortunately, the $\weakp$ operator satisfies a \emph{loop rule}
that simplifies reasoning about loops: if we can find an expectation $I$
satisfying an \emph{invariant} condition \fix{and some additional conditions},
then we can easily bound the $\weakp$
of a loop. Checking the invariant condition involves analyzing just the body of
the loop, rather than the entire loop. Thus, finding invariants is an important
obstacle towards automated reasoning about probabilistic programs.

\paragraph*{Discovering Invariants.}
Two prior works have considered how to automatically infer invariant
expectations for probabilistic loops. The first is
\textsc{Prinsys}~\cite{DBLP:conf/qest/GretzKM13}. Using a template with one
hole, \textsc{Prinsys} produces a first-order logical formula describing
possible substitutions satisfying the invariant condition. While effective for
their benchmark programs, the method's reliance on templates is limiting;
furthermore, the user must manually solve a system of logical formulas to find
the invariant.

The second work, by \citet{DBLP:conf/cav/ChenHWZ15}, focuses on inferring
polynomial invariants. They apply Lagrange interpolation theorem to find a polynomial
invariant. However, many invariants are not polynomials: for instance, an
invariant may combine two polynomials piecewise by branching on a Boolean
condition.

\paragraph*{Our Approach: Invariant Learning.}
We take a different approach inspired by data-driven invariant
learning for ``regular" programs~\cite{DBLP:conf/fm/FlanaganL01,daikon}.
In these methods, an invariant is seen as a classifier between a set of ``good" states that satisfy the specification, and a set of ``bad" states that violate the specification.
For training data, the program
is profiled to collect execution states. Then, the program is executed with a set of inputs to generate the training data. The invariant is synthesized using machine learning algorithm\fix{s} to
find a classifier between the ``good" and the ``bad" states". Data-driven techniques reduce the
reliance on templates, and can treat the program as a black box--- \fix{the} learner only needs to
execute the program to gather input and output data. But to extend the
data-driven method to the probabilistic setting, there exist significant challenges:
\begin{itemize}
	\item \textbf{Quantitative invariants.} While the logic of expectations
	resembles the logic of standard assertions, an important difference is that
	expectations are \emph{quantitative}: they map program states to real
	numbers, not a binary yes/no. While standard invariant learning is a
	\emph{classification} task (i.e., predicting a binary label given a program
	state), our probabilistic invariant learning is closer to a
	\emph{regression} task (i.e., predicting a number given a program state).
	\item \textbf{Stochastic data.} Standard invariant learning assumes the program
	behaves like a \emph{function}: a given input state always leads to the same
	output state. In contrast, a probabilistic program takes an input state to a
	distribution over outputs. Since we are only able to observe a single draw
	from the output distribution each time we run the program, execution traces
	in our setting are inherently \emph{noisy}. Accordingly, we cannot hope to
	learn an invariant that fits the observed data perfectly, even if the
	program has an invariant---our learner must be robust to noisy training
	data.
	\item \textbf{Complex learning objective.} To fit a probabilistic invariant to
	data, the logical constraints defining an invariant must be converted into a
	regression problem with a loss function suitable for standard machine
	learning algorithms and models. While typical regression problems relate the
	unknown quantity to be learned to known data, the conditions defining
	invariants are somehow self-referential: they describe how an unknown
	invariant must be related to itself. This feature makes casting invariant
	learning as machine learning a difficult task.
\item \textbf{Quality of examples.} If a candidate invariants fails the
  verification check, the generated counterexamples are added back to the
  data-set for learning a better invariant. However, if the generated
  counterexamples are ``close" to valid examples, the generated loss may not be
  enough to move the learning to a new invariant, thereby stalling progress. To
  enable progress, we pose the search for counterexamples as an optimization
  problem to search for the \textit{worst-case} counterexamples that would
  generate appreciable loss.
\end{itemize}

\paragraph*{Outline.}
After covering preliminaries (\cref{sec:prelim}) and stating our problem
(\cref{sec:problem}), we present our contributions.
\begin{itemize}
	\item A general method called \tool for learning invariants for
    \fix{almost surely terminating} probabilistic
	programs (\cref{sec:algorithm}). \tool executes the program multiple times on
	a set of input states, and then uses machine learning algorithms to learn
	models encoding possible invariants. A CEGIS-like loop is used to
	iteratively expand the dataset after encountering incorrect candidate
	invariants.
	\item Concrete instantiations of \tool tailored for handling two problems:
	learning \emph{exact invariants}~(\cref{sec:exactinv}), and learning
	\emph{sub-invariants}~(\cref{sec:subinv}). Our method for exact invariants
	learns a \emph{model tree}~\citep{quinlan1992learning}, a generalization of
	binary decision trees to regression. The constraints for sub-invariants are
	more difficult to encode as a regression problem, and our method learns a
	\emph{neural model tree}~\cite{DBLP:journals/corr/abs-1806-06988} with a
	custom loss function. While the models differ, both algorithms leverage
	off-the-shelf learning algorithms.
	\item An implementation of \tool and a thorough evaluation on a large set of
	benchmarks (\cref{sec:eval}). Our tool can learn invariants and
	sub-invariants for examples considered in prior work, and more difficult
	versions that are beyond the reach of prior work.
\end{itemize}
We discuss related work in \cref{sec:rw}.

\section{Preliminaries} \label{sec:prelim}

\paragraph*{Probabilistic Programs.}
We will consider programs written in \LANG, a basic probabilistic imperative
language with the following grammar:
\begin{align*}
	P \defdefeq \Skip
	\mid \Assn{x}{e}
	\mid \Rand{x}{d}
	\mid \Seq{P}{P}
	\mid \Cond{e}{P}{P}
	\mid \while{e}{P}
\end{align*}
Above, $x$ ranges over a countable set of variables $\mathcal{X}$, $e$ is an
expression, and $d$ is a distribution expression. Expressions are interpreted in
program states $\sigma : \mathcal{X} \to \mathcal{V}$, which map variables to a
set of values $\mathcal{V}$ (e.g., booleans, integers).
The semantics of probabilistic programs is
defined in terms of distributions. To avoid measure-theoretic technicalities, we
assume that $\mathcal{V}$ is countable.

\begin{definition}
	A (discrete) distribution $\dist$ over a countable set $S$ is a function of
	type $S \to \mathbb{R}_{\geq 0}$ satisfying $\sum_{s \in S}\dist(s) = 1$. We
	denote the set of distributions over $S$ by $\Dist{S}$.
\end{definition}

Let $\States$ denote the set of all program states
As is standard, programs $P$ are interpreted as maps $\denot{P} : \States \to
\Dist{\States}$. This definition requires two standard operations on
distributions.
\begin{definition}
	Given a set $S$, $\unit$ maps any $s \in S$ to the Dirac distribution on
	$s$, i.e., $\unit(s)(s') \defeq 1$ if $s = s'$ and $\unit(s)(s') \defeq 0$ if $s
	\neq s'$.

	Given $\dist \in \Dist{S}$ and a map $f: S \to \Dist{T}$, the map $\bind$
	combines them into a distribution over $T$, $\bind(\dist, f) \in \Dist{T}$,
	defined via
	\begin{align*}
		\bind(\dist, f)(t) \defeq \sum_{s \in S} \dist(s) \cdot f(s)(t).
	\end{align*}
\end{definition}
The full semantics is presented in \cref{fig:semantics}; we comment on a few
details here. First, given a state $\sigma$, we interpret expressions $e$ and
distribution expressions $d$ as values $\denot{e}_\sigma \in \mathcal{V}$ and
distributions over values $\denot{d}_\sigma \in \Dist{\mathcal{V}}$,
respectively; we implicitly assume that all expressions are well-typed. Second,
since the program semantics maps $\States$ to \emph{distributions} over $\States$,
the semantics for loops is only well-defined when the loop is \emph{almost
	surely terminating} (AST): from any initial state, the loop terminates with
probability $1$.
\fix{Since our data-driven procedure will require running probabilistic programs
on concrete inputs, \emph{we assume throughout that all loops are almost surely
terminating (AST)}. This condition can often be verified using existing methods~(e.g.,
	\citep{Chatterjee2016,Chatterjee:2016:AAQ:2837614.2837639,mciver2016new}).}

\begin{figure*}
	\begin{align*}
		\llbracket \Skip \rrbracket_{\sigma} &\defeq \unit(\sigma) \\
		\llbracket \Assn{x}{\expr} \rrbracket_{\sigma} &\defeq \unit(\sigma[x \mapsto \denot{e}_{\sigma}])\\
		\llbracket \Rand{x}{d} \rrbracket_{\sigma} &\defeq \bind(\llbracket d \rrbracket_{\sigma}, v \mapsto \unit(\sigma[x \mapsto v]) )\\
		\llbracket \Seq{P_1}{P_2} \rrbracket_{\sigma} &\defeq \bind(\llbracket P_1 \rrbracket_{\sigma}, \sigma' \mapsto \llbracket P_2 \rrbracket_{\sigma'}) \\
		\llbracket \Cond{e}{P_1}{P_2} \rrbracket_{\sigma} &\defeq
		\begin{cases}
			\llbracket P_1 \rrbracket_\sigma &: \denot{e}_\sigma =  \kwtrue \\
			\llbracket P_2 \rrbracket_\sigma &: \denot{e}_\sigma = \kwfalse
		\end{cases} \\
		\llbracket \while{e}{P} \rrbracket_{\sigma} &\defeq \lim_{n \to \infty} \llbracket (\Cond{e}{P}{\Skip})^n \rrbracket_{\sigma}
	\end{align*}
	\caption{Program semantics}
	\label{fig:semantics}
\end{figure*}

\paragraph*{Weakest Pre-expectation Calculus.}
Morgan and McIver's \emph{weakest pre-expectation calculus} reasons about
probabilistic programs by manipulating \emph{expectations}.
\begin{definition}
  \label{def:expectation}
	Denote the set of program states by $\States$.
	Define the set of expectations, $\Expt$, to be
	$\{ E \mid E : \States \to \mathbb{R}^{\infty}_{\geq 0}\}$.
	Define
	$E_1 \leq E_2 \quad\text{iff}\quad \forall \sigma \in \States: E_1(\sigma) \leq E_2(\sigma)$.
	The set $\Expt$ is a complete lattice.
\end{definition}
While expectations are technically mathematical functions from $\States$ to the
non-negative extended reals, for formal reasoning it is convenient to work with
a more restricted syntax of expectations~(see,
e.g.,~\cite{DBLP:journals/pacmpl/BatzKKM21}). We will often view numeric
expressions as expectations. Boolean expressions $b$ can also be converted
to expectations; we let $[b]$ be the expectation that maps states where $b$
holds to $1$, and other states to $0$. As an example of our notation, $[flip =
0] \cdot (x + 1)$, $x + 1$ are two expectations, and we have $[flip =
0] \cdot (x + 1) \leq x + 1$.

\begin{figure*}
	\begin{align*}
		\weakp(\Skip, E) &\defeq E\   \hspace{120pt} \weakp(\Assn{x}{\expr}, E) \defeq E[\expr/x] \\
		\weakp(\Rand{x}{d}, E) &\defeq \lambda \sigma. \sum_{v \in \mathcal{V}} \llbracket d \rrbracket_{\sigma} (v) \cdot E[v/x] (\sigma) \hspace{25pt}
		\weakp(\Seq{P}{Q}, E) \defeq \weakp(P, \weakp(Q,E))\\
		\weakp&(\Cond{e}{P}{Q}, E) \defeq [e] \cdot \weakp(P,E) + [\neg e] \cdot \weakp(Q, E) \\
		\weakp&(\while{e}{P}, E) \defeq \lfp(\lambda X.~[e] \cdot \weakp(P, X) + [\neg e] \cdot E)
	\end{align*}
	\caption{Morgan and McIver's weakest pre-expectation operator}
	\label{fig:wpe}
\end{figure*}

Now, we are ready to introduce Morgan and McIver's \textit{weakest
	pre-expectation transformer} $\weakp$. In a nutshell, this operator takes a
program $P$ and an expectation $E$ to another expectation $E'$, sometimes called
the \emph{pre-expectation}. Formally, $\weakp$ is defined in \cref{fig:wpe}. The
case for loops involves the least fixed-point ($\lfp$) of $\Phi_E^{\weakp} \defeq
\lambda X. ([e] \cdot \weakp(P, X) + [\neg e] \cdot E)$, the
\emph{characteristic function} of the loop with respect to
$\weakp$~\cite{KaminskiKMO16}. \fix{The characteristic function is
Scott-continuous on the complete lattice $\Expt$, so the least fixed-point
exists by the Kleene fixed-point theorem.}

The key property of the $\weakp$ transformer is that for any program $P$,
$\weakp(P,E)(\sigma)$ is the expected value of $E$ over the output distribution
$\llbracket P \rrbracket_{\sigma}$.
\begin{theorem}[See, e.g.,~\cite{KaminskiKMO16}]
	\label{psemantics}
	For any program $P$ and expectation $E \in \Expt$,
	$\weakp(P,E) = \lambda \sigma. \sum_{\sigma' \in \Sigma}  E(\sigma') \cdot \llbracket P \rrbracket_{\sigma}(\sigma')$
\end{theorem}
Intuitively, the weakest pre-expectation calculus provides a syntactic way to compute
the expected value of an expression $E$ after running a program $P$, except when
the program is a loop. For a loop, the least fixed point definition of
$\weakp(\while{e}{P}, E)$ is hard to compute.

\section{Problem Statement} \label{sec:problem}


Analogous to when analyzing the weakest pre-conditions of a loop, knowing a
loop \emph{invariant} or \emph{sub-invariant} expectation helps one to
bound the loop's weakest pre-expectations, but a (sub)invariant
expectation can be difficult to find. Thus, we aim to develop an algorithm to
automatically synthesize invariants and sub-invariants of probabilistic loops.
More specifically, our algorithm tackles the following two problems:

\begin{enumerate}
	\item \textbf{Finding exact invariants: }
	Given a loop $\while{G}{P}$ and an expectation $\post$ as input,
	we want to find an expectation $I$ such that
	\begin{align}
		I = \Phi_{\post}^{\weakp}(I) \defeq [G] \cdot \weakp(P, I) + [\neg G] \cdot \post . \label{eq:invreq}
	\end{align}
  Such an expectation $I$ is an \emph{exact invariant} of the loop with respect
  to $\post$.
	\item \textbf{Finding sub-invariants: }
	Given a loop $\while{G}{P}$ and expectations $\pre, \post$, we aim to
	learn an expectation $I$ such that
	\begin{align}
		I &\leq \Phi_{\post}^{\weakp}(I) \defeq [G] \cdot \weakp(P, I) + [\neg G] \cdot \post \label{eq:subinvreq} \\
		\pre &\leq I \label{eq:preI}.
	\end{align}
  The first inequality~\cref{eq:subinvreq} says that $I$ is a sub-invariant: on states that satisfy
  $G$, the value of $I$ lower bounds the expected value of itself after running
  one loop iteration from initial state, and on states that violate $G$, the
  value of $I$ lower bounds the value of $\post$. The second inequality~\cref{eq:preI} says that $I$ is lower-bounded by the given expectation $\pre$.
\end{enumerate}

Note that an exact invariant is a sub-invariant, so one indirect way to solve
the second problem is to solve the first problem, and then check $\pre \leq I$.
However, we aim to find a more direct approach to solve the second problem
because often exact invariants can be complicated and hard to find, while
sub-invariants can be simpler and easier to find.

\fix{Once we find (sub)-invariants for a loop, we can use the (sub)-invariants to
derive provable bounds on the weakest pre-expectation of the loop if the
(sub)-invariants satisfy some additional conditions. Prior work has identified
various sufficient conditions; we use the conditions identified
by~\citet{hark2019aiming} because they are relatively easy to check.  We use the
following corollary of~\citet[Theorem 38]{hark2019aiming}.}

\begin{proposition}\label{theorem:equivalence}
  \fix{
    When $E$ and $I$ are both expectations,
  if in addition \emph{one of} (a), (b) or (c) holds:
  \begin{enumerate}[(a)]
    \item The number of iterations that  $\while{G}{P}$ runs is bounded, and $(\Phi_E^{\weakp})^n(I)$ is finite for every $n \in \mathbb{N}$.
    \item The following four conditions are all satisfied:
      \begin{itemize}
        \item The expected looping time of $\while{G}{P}$ is finite for every initial state $s \in \States$,
        \item $\Phi_E^{\weakp}(I)$ is finite.
        \item There exists an expectation $I'$ such that $I = [\neg e] \cdot E +
          [e] \cdot I$.
        \item The conditional difference of the invariant $I$, i.e., $\Delta I
          \defeq \lambda s. ([e] \cdot \weakp(P, |I - I(s)|))(s)$ is bounded by
          a constant.
      \end{itemize}
    \item $\while{G}{P}$ is almost surely terminating and
      both $I$ and $E$ are bounded.
  \end{enumerate}
   then we have:
	\begin{align}
    I \leq  \Phi_{E}^{\weakp}(I)  &\implies I \leq \weakp(\while{G}{P}, E) \label{eq:subinv}\\
    \text{ and } I = \Phi_{E}^{\weakp}(I) &\implies I = \weakp(\while{G}{P}, E) . \label{eq:equivinv}
	\end{align}
}
\end{proposition}
\begin{proof}
\fix{
  For~\cref{eq:subinv}, note that our conditions and the claim are exactly the
  same as those in~\citet[Theorem 38]{hark2019aiming}.
}

\fix{For~\cref{eq:equivinv}, the definition of the weakest pre-condition
  operator $\weakp(\while{G}{P}, E) = \lfp \Phi_{E}^{\weakp}$ and the Park
  induction principle $I \geq \Phi_{E}^{\weakp}(I) \implies I \geq \lfp
  \Phi_{E}^{\weakp}$~\citep{park1969fixpoint} gives:
  \begin{align}
    I \geq \Phi_{E}^{\weakp}(I) \implies I \geq \weakp(\while{G}{P}, E) . \label{eq:park}
  \end{align}
  Combining this implication~\cref{eq:park} with~\cref{eq:subinv}, we get~\cref{eq:equivinv}.
}
\end{proof}

\section{Algorithm}
\label{sec:algorithm}

\begin{figure}
	\centering
	\includegraphics[scale=0.4]{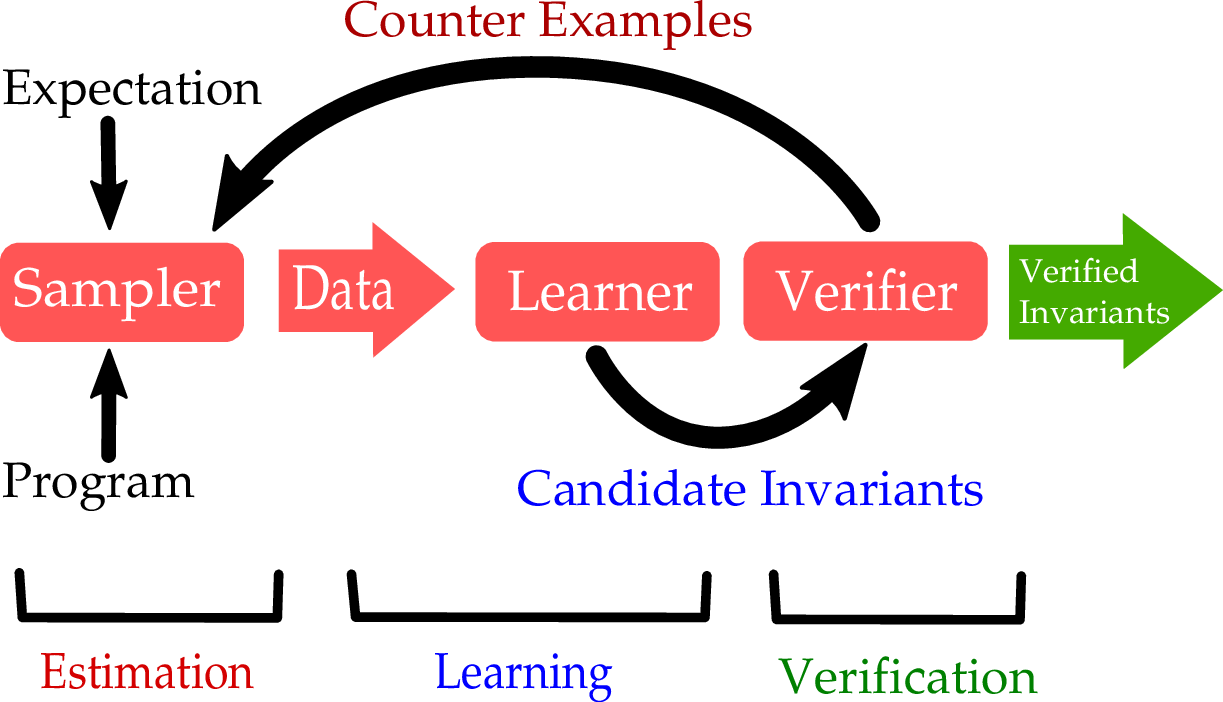}
	\caption{Overview of \tool}
	\label{fig:overview}
\end{figure}

We solve both problems with one algorithm, \tool (short for EXpectation
Invariant SynThesis). Our data-driven method resembles Counterexample Guided
Inductive Synthesis (CEGIS) (see~\cref{fig:overview}), but differs in two ways.
First, candidates are synthesized by fitting a machine learning model to data consisted of
program traces starting from random input states. Our target programs are also
probabilistic, introducing a second source of randomness to program traces.
Second, our approach seeks high-quality counterexamples---violating the target
constraints as much as possible---in order to improve synthesis. For
synthesizing invariants and sub-invariants, such counterexamples can be
generated by using a computer algebra system to solve an optimization problem.

In this section, we introduce a meta-algorithm to tackle both problems discussed in ~\cref{sec:problem}. We will see how to instantiate the meta-algorithm's subroutines in~\cref{sec:exactinv} and~\cref{sec:subinv}.

\begin{figure}
	\begin{align*}
		& \tool(\pprog, \pexp, \nruns, \nstates) \textbf{:} \\
		& \qquad \feat \gets \getfeatures(\pprog, \pexp) \\
		& \qquad \states \gets \getstates(\feat, \nstates) \\
		& \qquad \mathit{data} \gets \sample(\pprog, \pexp, \feat, \nruns, \states) \\
		& \qquad \textbf{while} \text{ not timed out} \textbf{:} \\
		& \qquad \qquad models \gets \learninv(\feat, \mathit{data})\\
		& \qquad \qquad \mathit{candidates} \gets \extractinv(models)\\
		& \qquad \qquad \textbf{for } inv \textbf{ in } \mathit{candidates} \textbf{:} \\
		& \qquad \qquad \qquad \mathit{verified}, \mathit{cex} \gets \verifyinv(inv, \pprog)\\
		& \qquad\qquad \qquad \textbf{if } \mathit{verified} \textbf{:}\\
		& \qquad\qquad \qquad\qquad \textbf{return } inv \\
		& \qquad\qquad \qquad \textbf{else:} \\
		& \qquad\qquad\qquad\qquad  \states \gets \states \cup cex \\
		& \qquad\qquad\qquad\qquad  \states \gets \states \cup \getstates(\feat, \nstates') \\
		& \qquad\qquad \mathit{data} \gets \mathit{data} \cup \sample(\pprog, \pexp, \feat, nruns, \states)
	\end{align*}
	\caption{Algorithm \tool \label{fig:cegl}}
\end{figure}

We present the pseudocode in~\cref{fig:cegl}. \tool takes a probabilistic
program $\pprog$, a post-expectation or a pair of pre/post-expectation
$\pexp$, and hyper-parameters $\nruns$ and $\nstates$. \tool \fix{starts} by
generating a list of features $\feat$, which are numerical expressions formed
by program variables used in $\pprog$. Next, \tool samples $\nstates$
initialization $\states$ and runs $\pprog$ from each of those states for
$\nruns$ trials, and records the value of $\feat$ on program traces as
$\mathit{data}$.  Then, \tool enters a CEGIS loop.  In each iteration of the
loop, first the learner $\learninv$ trains models to minimize their violation
of the required inequalities (e.g., ~\cref{eq:subinvreq} and~\cref{eq:preI}
for learning sub-invariants) on $\mathit{data}$.  Next, $\extractinv$
translates learned models into a set $\mathit{candidates}$ of expectations.
For each candidate $inv$, the verifier $\verifyinv$ looks for program states
that \emph{maximize} $inv$'s violation of required inequalities. If it cannot
find any program state where $inv$ violates the inequalities, the verifier
returns $inv$ as a valid invariant or sub-invariant. Otherwise, it produces a
set $\mathit{cex}$ of counter-example program states, which are added to the
set of initial states. Finally, before entering the next iteration, the
algorithm augments $\states$ with a new batch of $\nstates'$ initial states,
generates trace data from running $\pprog$ on each of these states for
$\nruns$ trials, and augments the dataset $\mathit{data}$.  This data
augmentation ensures that the synthesis algorithm collects more and more
initial states, some randomly generated ($\getstates$) and some from prior
counterexamples ($\mathit{cex}$), guiding the learner towards better
candidates. Like other CEGIS-based tools, our method is sound but not
complete, i.e., if the algorithm returns an expectation then it is guaranteed
to be an exact invariant or sub-invariant, but the algorithm might never
return an answer; in practice, we set a timeout.

\section{Learning Exact Invariants} \label{sec:exactinv}

In this section, we detail how we instantiate \tool's subroutines to learn
an exact invariant $I$ satisfying $I = \Phi_{\post}^{\weakp}(I)$, given a
loop $\pprog$ and an expectation $\pexp = \post$.

At a high level, we first sample a set of program states $\states$
using $\getstates$. From each program state $s \in \states$,
\sample executes $\pprog$ and estimates $\weakp(\pprog,
\post)(s)$. Next, \learninv trains regression models $M$ to predict
the estimated $\weakp(\prog, \post)(s)$ given the value of features
evaluated on $s$. Then, \extractinv translates the learned models $M$ to an
expectation $I$.  In an ideal scenario, this $I$ would be equal to
$\weakp(\prog, \post)$, which is also always an exact invariant.
But since $I$ is learned from stochastic data, it may be noisy. So, we use
\verifyinv to check whether $I$ satisfies the invariant condition $I =
\Phi_{\post}^{\weakp}(I)$.

The reader may wonder why we took this complicated approach, first
estimating the weakest pre-expectation of the loop, and then computing the
invariant: If we are able to learn an expression for $\weakp(\prog, \post)$
directly, then why are we interested in the invariant $I$? The answer is
that \fix{it is easier to verify if an $I$ is an invariant than to check whether
and $I$ is the least fixed point. Once we verify that $I$ is an invariant that
additionally satisfies conditions in~\cref{theorem:equivalence}, then we also know
that $I = \weakp(\prog, \post)$. }
Since our learning process is
inherently noisy, this verification step is crucial and motivates why we
want to find an invariant.


\begin{figure}
	\centering
	\begin{subfigure}[b]{0.3\textwidth}
		\begin{align*}
			&\while{x = 0}{} \\
			&\qquad \Assn{n}{n + 1}; \\
			&\qquad \Rand{x}{\mathbf{Bernoulli}(p)}
		\end{align*}
		\caption{Program: $\prog$}
		\label{fig:prog}
	\end{subfigure}
	\hfill
	\begin{subfigure}[b]{0.35\textwidth}
		\begin{center}
			\tikzset{every picture/.style={line width=0.75pt}} 

			\begin{tikzpicture}[x=0.75pt,y=0.75pt,yscale=-0.8,xscale=0.8]

				\draw    (180,50) -- (141.75,90.87) ;
				\draw [shift={(140.5,92.43)}, rotate = 308.63] [color={rgb, 255:red, 0; green, 0; blue, 0 }  ][line width=0.75]    (10.93,-3.29) .. controls (6.95,-1.4) and (3.31,-0.3) .. (0,0) .. controls (3.31,0.3) and (6.95,1.4) .. (10.93,3.29)   ;
				\draw    (200,50) -- (240,90.87) ;
				\draw [shift={(241.2,92.43)}, rotate = 230.67000000000002] [color={rgb, 255:red, 0; green, 0; blue, 0 }  ][line width=0.75]    (10.93,-3.29) .. controls (6.95,-1.4) and (3.31,-0.3) .. (0,0) .. controls (3.31,0.3) and (6.95,1.4) .. (10.93,3.29)   ;

				\draw (170,31.4) node [anchor=north west][inner sep=0.75pt]    {$x=0?$};
				\draw (120,54) node [anchor=north west][inner sep=0.75pt]    {$x \neq 0$};
				\draw (220,53) node [anchor=north west][inner sep=0.75pt]    {$x = 0$};
				\draw (130,99.4) node [anchor=north west][inner sep=0.75pt]    {$n$};
				\draw (215,95) node [anchor=north west][inner sep=0.75pt]    {$n+\frac{1}{p} $};
			\end{tikzpicture}
		\end{center}
		\caption{Model tree for $\weakp(\prog, n)$}
		\label{fig:wpe_model_tree}
	\end{subfigure}
	\hfill
	\begin{subfigure}[b]{0.3\textwidth}
		\begin{center}
			\tikzset{every picture/.style={line width=0.75pt}} 

			\begin{tikzpicture}[x=0.75pt,y=0.75pt,yscale=-0.8,xscale=0.8]

				\draw    (180,50) -- (141.75,90.87) ;
				\draw [shift={(140.5,92.43)}, rotate = 308.63] [color={rgb, 255:red, 0; green, 0; blue, 0 }  ][line width=0.75]    (10.93,-3.29) .. controls (6.95,-1.4) and (3.31,-0.3) .. (0,0) .. controls (3.31,0.3) and (6.95,1.4) .. (10.93,3.29)   ;
				\draw    (200,50) -- (240,90.87) ;
				\draw [shift={(241.2,92.43)}, rotate = 230.67000000000002] [color={rgb, 255:red, 0; green, 0; blue, 0 }  ][line width=0.75]    (10.93,-3.29) .. controls (6.95,-1.4) and (3.31,-0.3) .. (0,0) .. controls (3.31,0.3) and (6.95,1.4) .. (10.93,3.29)   ;

				\draw (170,31.4) node [anchor=north west][inner sep=0.75pt]    {$x=0?$};
				\draw (120,54) node [anchor=north west][inner sep=0.75pt]    {$x \neq 0$};
				\draw (220,53) node [anchor=north west][inner sep=0.75pt]    {$x = 0$};
				\draw (130,99.4) node [anchor=north west][inner sep=0.75pt]    {$n$};
				\draw (210,95) node [anchor=north west][inner sep=0.75pt]    {$n+0.95\frac{1}{p} $};
			\end{tikzpicture}
		\end{center}
		\caption{Another model tree}
		\label{fig:model_tree}
	\end{subfigure}
	\caption{Running example: Program and model tree}
	\label{fig:example}
\end{figure}

\paragraph{A running example.}
We will illustrate our approach using~\cref{fig:example}. The simple program
\prog repeatedly loops: whenever $x$ becomes non-zero we exit the loop;
otherwise we increase $n$ by $1$ and draw $x$ from a biased coin-flip
distribution ($x$ gets $1$ with probability $p$, and $0$ otherwise). We aim to
learn $\weakp(\prog, n)$, which is $[x \neq 0] \cdot n + [x = 0] \cdot (n + \frac{1}{p})$.

\paragraph*{Our Regression Model.}

Before getting into how \tool collects data and trains models, we introduce the
class of regression models it uses -- \emph{model trees}, a generalization of
decision trees to regression tasks~\citep{quinlan1992learning}. Model trees are
naturally suited to expressing piecewise functions of inputs, and are
straightforward to train. While our method can in theory generalize to other
regression models, our implementation focuses on model trees.

More formally, a model tree $\Tree \in \mathcal{T}$ over
features $\mathcal{F}$ is a full binary tree where each internal node is
labeled with a predicate $\phi$ over variables from $\mathcal{F}$, and each
leaf is labeled with a real-valued model $M \in \mathcal{M} :
\mathbb{R}^\mathcal{F} \to \mathbb{R}$.
Given a feature vector in $x \in \mathbb{R}^\mathcal{F}$, a model tree $\Tree$
over $\mathcal{F}$ produces a
numerical output $\Tree(x) \in \mathbb{R}$ as follows:
\begin{itemize}
	\item If $\Tree$ is of the form $\Leaf(M)$, then $\Tree(x) \defeq M(x)$.
	\item If $\Tree$ is of the form $\Node(\phi, \Tree_L, \Tree_R)$, then
	$\Tree(x) \defeq \Tree_R(x)$ if the predicate $\phi$ evaluates to true on $x$,
	and $\Tree(x) \defeq \Tree_L(x)$ otherwise.
\end{itemize}

Throughout this paper, we consider model trees of the following form as our
regression model. First, node predicates $\phi$ are of the form $f \bowtie c$,
where $f \in \mathcal{F}$ is a feature, ${\bowtie} \in {\{ < , \leq, =, >, \geq
	\}}$ is a comparison, and $c$ is a numeric constant. Second, leaf models on a
model tree are either all \emph{linear models} or all products of constant
powers of features, which we call \emph{multiplication models}.  For example,
assuming $n, \frac{1}{p}$ are both features,~\cref{fig:wpe_model_tree} and
~\cref{fig:model_tree} are two model trees with linear leaf models,
and~\cref{fig:wpe_model_tree} expresses the weakest pre-expectation
$\weakp(\prog, n)$.  Formally, the leaf model $M$ on a feature vector $f$ is
either
\[
M_l(f) = \sum_{i = 1}^{|\mathcal{F}|} \alpha_i \cdot f_i
\qquad\text{ or }\qquad
M_m(f)	= \prod_{i = 1}^{|\mathcal{F}|} f_i^{\alpha_i}
\]
with constants $\{ \alpha_i \}_i$.
Note that multiplication models can also be viewed as linear models on
logarithmic values of features because $\log M_m(f) = \sum_{i = 1}^{|\mathcal{F}|} \alpha_i
\cdot \log(f_i)$.
While it is also straightforward to adapt our method to other leaf models, we
focus on linear models and multiplication models because of their simplicity and
expressiveness. Linear models and multiplication models also complement
each other in their expressiveness: encoding expressions like $x + y$ uses
simpler features with linear models (it suffices if $\mathcal{F} \ni x, y$, as
opposed to needing $\mathcal{F} \ni x + y$ if using multiplicative models),
while encoding $\frac{p}{1-p}$ uses simpler features with multiplicative models
(it suffices if $\mathcal{F} \ni p, 1-p$, as opposed to needing $\mathcal{F} \ni
\frac{p}{1-p}$ if using linear models).


\subsection{Generate Features (\getfeatures)}
\label{sec:inv:getfeatures}
Given a program, the algorithm first generates a set of features $\mathcal{F}$
that model trees can use to express unknown invariants of the given loop.  For
example, for \prog, $I = [x \neq 0] \cdot n + [x = 0] \cdot (n+ \frac{1}{p})$
is an invariant, and to have a model tree (with linear/multiplication leaf
models) express $I$, we want $\mathcal{F}$ to include both $n$ and
$\frac{1}{p}$, or $n + \frac{1}{p}$ as one feature.  $\mathcal{F}$ should
include the program variables at a minimum, but it is often useful to have
more complex features too.  While generating more features increases the
expressivity of the models, and richness of the invariants, there is a cost:
the more features in $\mathcal{F}$, the more data is needed to train a model.

Starting from the program variables, \getfeatures generates two lists of
features, $\mathcal{F}_{l}$ for linear leaf models and $\mathcal{F}_{m}$ for
multiplication leaf models. Intuitively, linear models are more expressive if
the feature set $\mathcal{F}$ includes some products of terms, e.g., $n \cdot
p^{-1}$,  and multiplication models are more expressive if $\mathcal{F}$
includes some sums of terms, e.g., $n + 1$.

We assume program variables and optional user-supplied features $\mathit{opt}$ are typed as probabilities (denoted using $p_i$), integers (denoted using $n_i$), booleans (denoted using $b_i$), or reals (denoted using $x_i$).
\fix{In general, we do not restrict the integers $n_i$ and reals $x_i$ to be non-negative as our learning algorithm does not assume non-negativity; later, though, to ensure that what \tool generates are indeed expectations according to~\cref{def:expectation}, 
we assume variables in specific programs to be non-negative. }
Then, given program variables and
user-supplied features $p_i, \dots, n_i, \dots, b_i, \dots, x_i, \dots$,
a loop with guard $G$, and post expectation $\mathit{pexp}$,
\getfeatures generates
\begin{align*}
	\mathcal{F}_l &
	\ni G \mid pexp
	\mid p_i \mid n_i \mid b_i \mid x_i
	\mid p_i \cdot p_j
	\mid n_i \cdot n_j
	\mid x_i \cdot x_j
	\mid n_i \cdot x_j
	\mid b_i \cdot b_j \\
	\mathcal{F}_m &
	\ni G \mid pexp
	\mid p_i \mid n_i \mid b_i \mid x_i
	\mid 1+ p_j
	\mid 1 - p_j
	\mid p_i + p_j
	\mid p_i + p_j - (p_i \cdot p_j) \\
	& \qquad \mid n_i + n_j
	\mid n_i - n_j
	\mid x_i + x_j
	\mid x_i - x_j
	\mid n_i + x_j
	\mid n_i - x_j
	\mid b_i + b_j
	\mid b_i - b_j.
\end{align*}

\subsection{Sample Initial States (\getstates)}
Recall that \tool aims to learn an expectation $I$ that is equal to the
weakest pre-expectation $\weakp(\while{G}{P}, \post)$.  A natural idea for
\sample is to run the program from all possible initializations multiple
times, and record the average value of $\post$ from each initialization. This
would give a map close to 	$\weakp(\while{G}{P}, \post)$ if we run enough
trials so that the empirical mean is approximately the actual mean. However,
this strategy is clearly impractical---many of the programs we consider have
infinitely many possible initial states (e.g., programs with integer
variables). Thus, \getstates needs to choose a manageable number of initial
states for \sample to use.

In principle, a good choice of initializations should exercise as many parts
of the program as possible. For instance, for \prog in~\cref{fig:example}, if
we only try initial states satisfying $x \neq 0$, then it is impossible to learn
the term $[x = 0] \cdot (n + \frac{1}{p})$ in $\weakp(\prog, n)$
from data. However, covering the control flow graph may not be enough. Ideally,
to learn how the expected value of $\post$ depends on the
initial state, we also want data from multiple initial states along
each path.

While it is unclear how to choose initializations to ensure optimal coverage,
our implementation uses a simpler strategy: \getstates generates $\nstates$
states in total, each by sampling the value of every program variable
uniformly at random from a space. We assume program variables are typed as
booleans, integers, probabilities, or floating point numbers and sample
variables of some type from the corresponding space. For boolean variables,
the sampling space is simply $\{0, 1\}$; for probability variables, the space
includes reals in some interval bounded away from $0$ and $1$, because
probabilities too close to 0 or 1 tend to increase the variance of programs
(e.g., making some loops iterate for a very long time); for floating point
number and integer variables, the spaces are respectively reals and integers
in some bounded range. This strategy, while simple, is already very effective
in nearly all of our benchmarks (see~\cref{sec:eval}), though other strategies
are certainly possible (e.g., performing a grid search of initial states from
some space).

\subsection{Sample Training Data (\sample)}

We gather training data by running the given program $\pprog$ on the set of
initializations generated by \getstates.
From each program state $s \in \states$, the subroutine
\sample runs $\prog$  for $\nruns$ times to get output states
$\{ s_1, \dots, s_{\nruns}  \}$	and
produces the following \fix{training example $(s, v)$, where
\begin{align*}
  v 	=	 \frac{1}{\nruns} \sum_{i = 1}^{\nruns} \post(s_i).
\end{align*}
}
%
Thus, the value $v$ is the empirical mean of $\post$ in the output state of
running $\prog$ from initial state $s_i$; as $\nruns$ grows
large, this average value approaches the true expected value
$\weakp(\prog, \post)(s)$.

\subsection{Learning a Model Tree (\learninv)}
Now that we have the training set $\mathit{data} = \{ (s_1, v_1),
\dots, (s_K, v_K) \}$ (where $K = \nstates$), we want to fit a model
tree $T$ to the data. We aim to apply off-the-shelf tools
that can learn model trees with customizable leaf models and loss.  For each
data entry, $v_i$ approximates $\weakp(\prog, \post)(s_i)$, so a
natural idea is to train a model tree $T$ that takes the value of features
on $s_i$ as input and predicts $v_i$. To achieve that,
we want to define the loss to measure
the error between predicted values $T(\mathcal{F}_l(s_i))$ (or
$T(\mathcal{F}_m(s_i))$) and the target value $v_i$.  Without loss of
generality, we can assume our invariant $I$ is of the form
\begin{align}
	I = \post + [G] \cdot I' \label{eq:template}
\end{align}
because $I$ being an invariant means
\begin{align*}
	I &= [\neg G] \cdot \post + [G] \cdot \weakp(P, I)
	= \post + [G] \cdot ( \weakp(P, I) - \post).
\end{align*}
In many cases, the expectation $I' = \weakp(P, I) - \post$ is simpler than $I$:
for example, the weakest pre-expectation of \prog can be expressed as $n + [x =
0] \cdot ( \frac{1}{p}) $; while $I$ is represented by a tree that splits on
the predicate $[x = 0]$ and needs both $n, \frac{1}{p}$ as features, the
expectation $I' = \frac{1}{p}$ is represented by a single leaf model tree that
only needs $p$ as a feature. \fix{Also, since we are rewriting $I$ in terms of
  $I'$ and $\textsf{post}$ only for the convenience of model-fitting, the $I'$
  here does not have to be an expectation, i.e., it could map states to
negative values. }

Aiming to learn weakest pre-expectations $I$ in the form of~\cref{eq:template},
\tool trains model trees $T$ to fit $I'$. More precisely,
\learninv trains a model tree $T_l$ with linear leaf models over features $\mathcal{F}_l$ by minimizing the loss
\begin{align}
	err_l(T_l, \mathit{data}) = \left( \sum_{i = 1}^{K}  \left(\post(s_i) + G(s_i) \cdot T_l(\mathcal{F}_l(s_i)) - v_i \right)^2 \right)^{1/2},
	\label{eq:errl}
\end{align}
where $\post(s_i)$ and $G(s_i)$ represents the value of expectation $\post$
and $G$ evaluated on the state $s_i$.  This loss measures the sum error
between the prediction $\post(s_i) + G(s_i) \cdot T_l(\mathcal{F}_l(s_i))$
and target $v_i$. Note that when the guard $G$ is false on an initial state
$s_i$, the example contributes zero to the loss because $\post(s_i) +
G(s_i) \cdot T_l(\mathcal{F}_l(s_i)) = \post(s_i) = v_i$; thus, we only need
to generate and collect trace data for initial states where the guard $G$ is
true.

Analogously,
\learninv trains a model tree $T_m$ with
multiplication leaf models  over features $\mathcal{F}_m$
to minimize the loss $err_m(T_m, data)$,
which is the same as $err_l(T_l, data)$ except $T_l(\mathcal{F}_l(s_i))$
is replaced by $T_m(\mathcal{F}_m(s_i))$ for each $i$.

	\subsection{Extracting Expectations from Models (\extractinv)}

	Given the learned model trees $T_l$ and $T_m$, we extract expectations that approximate $\weakp(\prog, \mathit{\post})$ in three steps:
	\begin{enumerate}
		\item \textbf{Round $T_l$, $T_m$ with different precisions.}
		Since we obtain the model trees $T_l$ and $T_m$ by learning and the training data is stochastic, the coefficients of features in $T_l$ and $T_m$ may be slightly off, so we apply several rounding schemes to generate a list of rounded model trees.

		For $T_l$, the learned model tree with linear leaf models, we round its coefficients to integers, one digit, and two digits respectively and
		get $T_{l0}, T_{l1},T_{l2}$. For instance, rounding the model tree depicted
		in~\cref{fig:model_tree} to integers gives us the model tree in~\cref{fig:wpe_model_tree}.
		For $T_m$, the learned model tree with multiplication leaf models
		, we construct $T_{m0}, T_{m1}, T_{m2}$ by rounding the leading constant
		coefficient to respective digits and all other exponentiating coefficients
		to integers: $c \cdot \prod_{i = 1}^{|\mathcal{F}|} x_i^{a_i}$
		gets rounded to $round(c, \mathit{digit}) \cdot \prod_{i = 1}^{|\mathcal{F}|} x_i^{\text{int}(a_i)}$.
		\item \textbf{Translate into expectations.} Since we learn model trees, this
		step is straightforward: for example, $n + \frac{1}{p}$ can be seen as a
		model tree (with only a leaf) mapping the values of features $n,\frac{1}{p}$
		to a number, or an expectation mapping program states where $n, p$ are
		program variables to a number.
		We translate each model tree obtained from the previous step to an expectation.
		\item \textbf{Form the candidate invariant.} Since we train the model trees to
		fit $I'$ so that $\post + [G] \cdot I'$ approximates $\weakp(\while{G}{P},
		\post)$, we construct each candidate invariant $inv \in \mathit{invs}$ by
		replacing $I'$ in the pattern $\post + [G] \cdot I'$ by an expectation
    obtained in the second step.
	\end{enumerate}

	\subsection{Verify Extracted Expectations (\verifyinv)}
	Recall that $\prog$ is a loop $\while{G}{P}$,
	and given a set of candidate invariants $\mathit{invs}$,
	we want to check if any $\mathit{inv} \in \mathit{invs}$ is a loop invariant,
	i.e., if $\mathit{inv}$ satisfies
	\begin{align}
		\mathit{inv} =  [\neg G] \cdot \post + [G] \cdot \weakp(P, \mathit{inv}).
		\label{eq:req:exactinv}
	\end{align}
	Since the learned model might not predict the expected value for every data
	point exactly, we must verify whether $inv$ satisfies this equality using
	$\verifyinv$.

\subsection{Search for \textit{worst-case} counterexamples}
If the invariant is not verified, \verifyinv has to look for counterexamples that violate the conditions for valid invariants. These counterexamples are fed back to augment the dataset to (hopefully) learn a better invariant.

    In this process, we hope that our data augmentation leads to \textit{substantial} loss in the learning algorithm so as to steer the learning to a new invariant. We provide an example on the challenges of such an endeavor. Though our algorithm uses regression, let us use classification for ease of discussion and visualization. Figure~\ref{fig:worsta} shows the case of learning a linear classifier: we are attempting to learn a classifier that separates the red and gray regions; the stars and circles are the data-points corresponding to the red and grey regions, respectively. Given the current dataset, we may learn an (incorrect) classifier, as shown by the dotted line.

    Figure~\ref{fig:worstb} shows the case where a new counterexample is found (counterexamples are shown with green boundary), that is used to augment the dataset. As the counterexample is quite close to the current classifier's decision boundary, it may not create enough loss to bulge the decision boundary. There are two possible solutions to this problem:

    \begin{itemize}
        \item \textbf{Engineer a classifier loss function that is more sensitive to violations. } We may engineer the loss function to penalize violations heavily. However, as our dataset is generated from estimates from a finite set of observations, there exist some amount of noise in the estimates. As our dataset is itself noisy, this option is not desirable as we would like our learnt classifier to be robust to noise. Further, such loss functions that are too sensitive to small violations can lead to instability in the learning process, making it difficult to converge.
        \item \textbf{Improve data augmentation. } Instead, we attempt to improve our data augmentation process in two ways: firstly, to attempt generate \textit{better} datapoints for augmentation by searching for \textit{worst-case} counterexamples, i.e. counterexamples that generate a highest possible loss on the objective function used to train the regression model (see Figure~\ref{fig:worstc}). Secondly, we generate a set of counterexamples so that the \textit{cumulative} loss is high enough to ensure progress towards the desirable invariant (see Figure~\ref{fig:worstd}).
    \end{itemize}

Hence, instead of searching for any counterexample that causes a violation, \verifyinv searches for a set of counterexamples that maximizes the
	violation in order to drive the learning process forward in the next iteration.
	Formally, for every $\mathit{inv} \in \mathit{invs}$, \verifyinv queries
	computer algebra systems to find a set of program states $S$ such that $S$
	includes states maximizing the absolute difference of two sides
	in~\cref{eq:req:exactinv}:
	\begin{align*}
		S \ni \mathbf{argmax}_s |\mathit{inv}(s) - \left( [\neg G] \cdot \post + [G] \cdot wp(P, \mathit{inv}) \right) (s) |.
	\end{align*}
	If there are no program state where the absolute difference is non-zero,
	\verifyinv returns $\mathit{inv}$ as a true invariant. Otherwise, the maximizing
	states in $S$ are added to the list of counterexamples $cex$; if no candidate in
	$\mathit{invs}$ is verified, \verifyinv returns \textsf{False} and the
	accumulated list of counterexamples $cex$. The next iteration of the CEGIS loop
	will sample program traces starting from these counterexample initial states,
	hopefully leading to a learned model with less error.

	\begin{figure}[t]
		\centering

		\begin{subfigure}[b]{0.4\textwidth}
			\includegraphics[width=\textwidth]{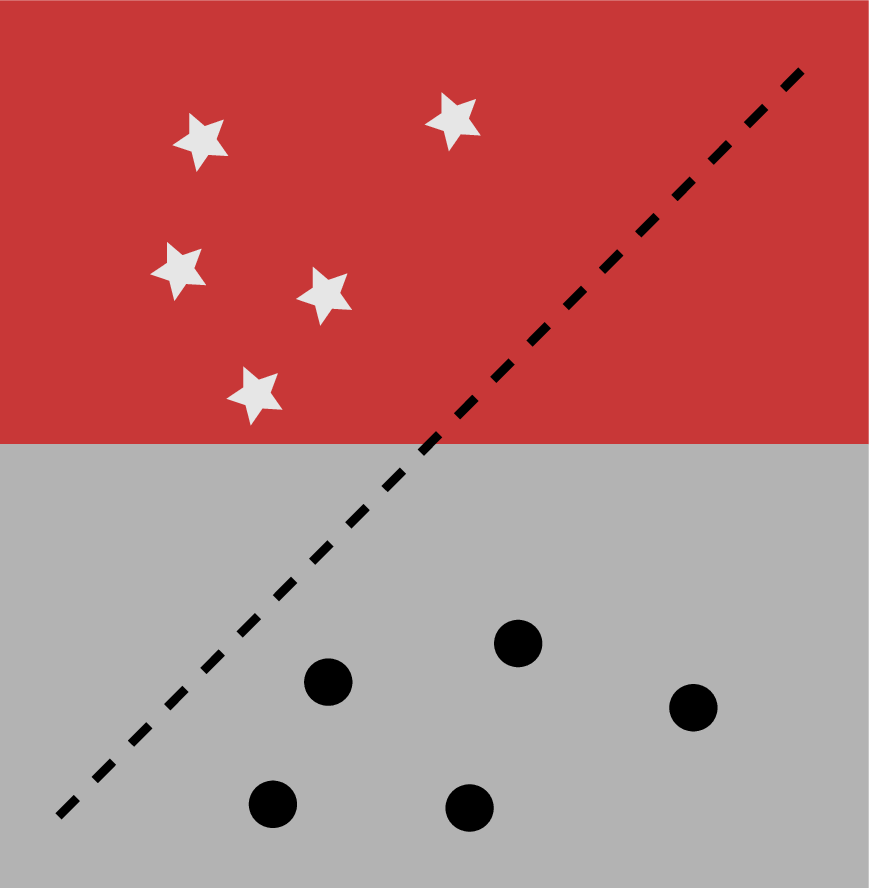}
			\caption{}
			\label{fig:worsta}
		\end{subfigure}
		\begin{subfigure}[b]{0.4\textwidth}
			\includegraphics[width=\textwidth]{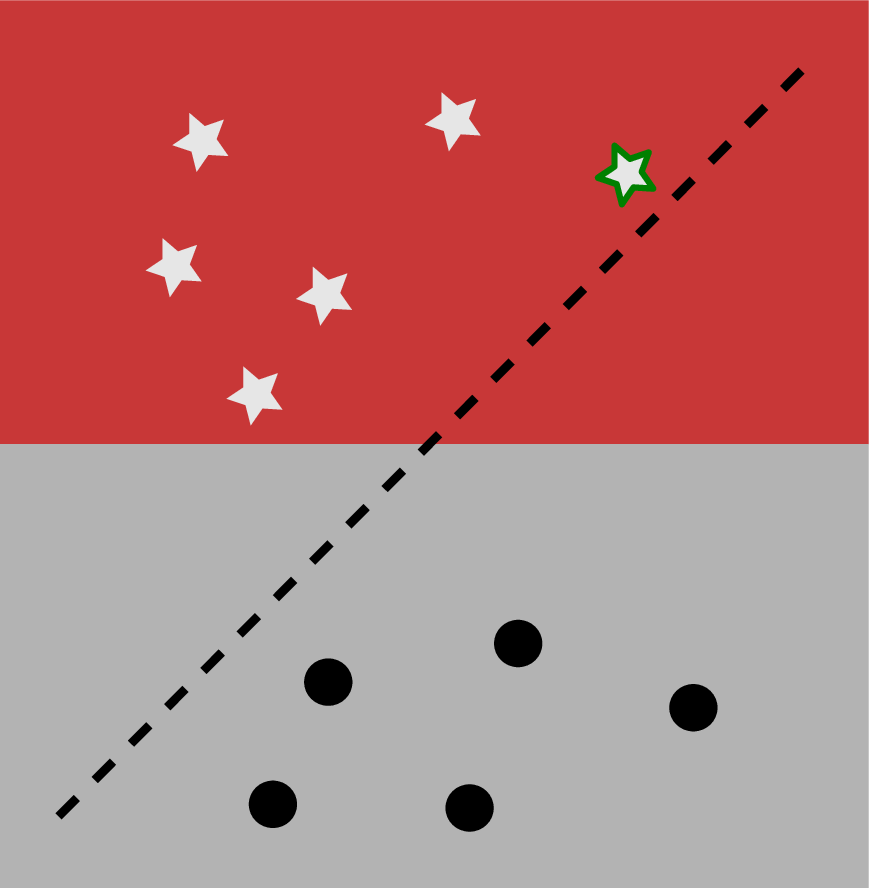}
			\caption{}
			\label{fig:worstb}
		\end{subfigure}

		\begin{subfigure}[b]{0.4\textwidth}
			\includegraphics[width=\textwidth]{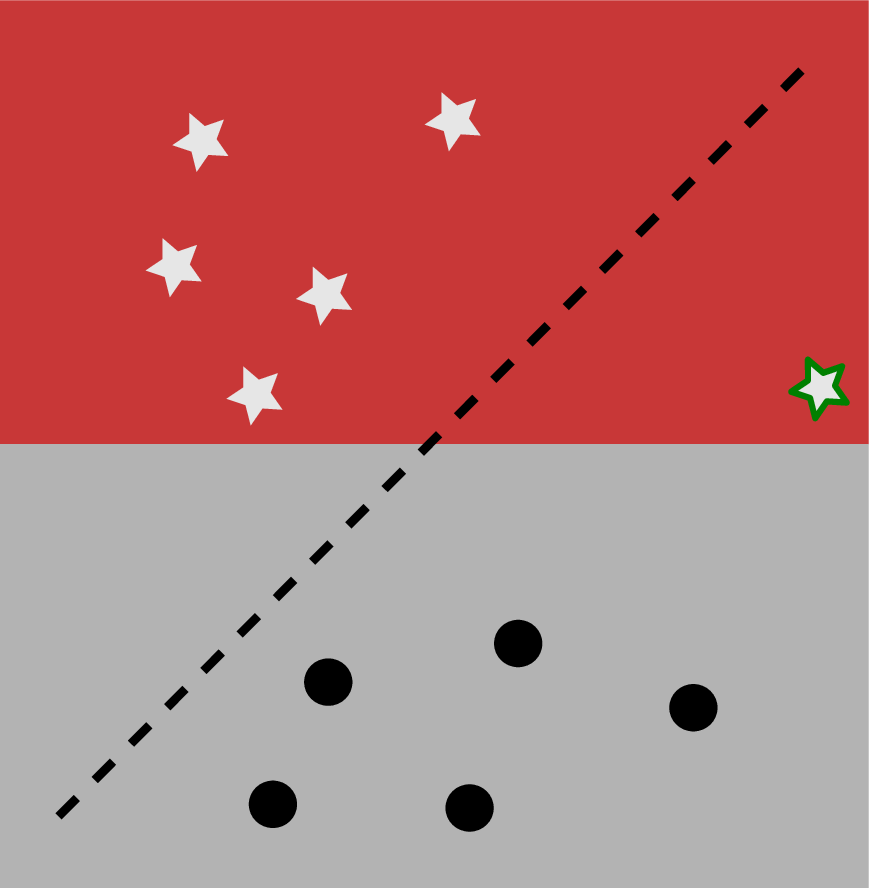}
			\caption{}
			\label{fig:worstc}
		\end{subfigure}
		\begin{subfigure}[b]{0.4\textwidth}
			\includegraphics[width=\textwidth]{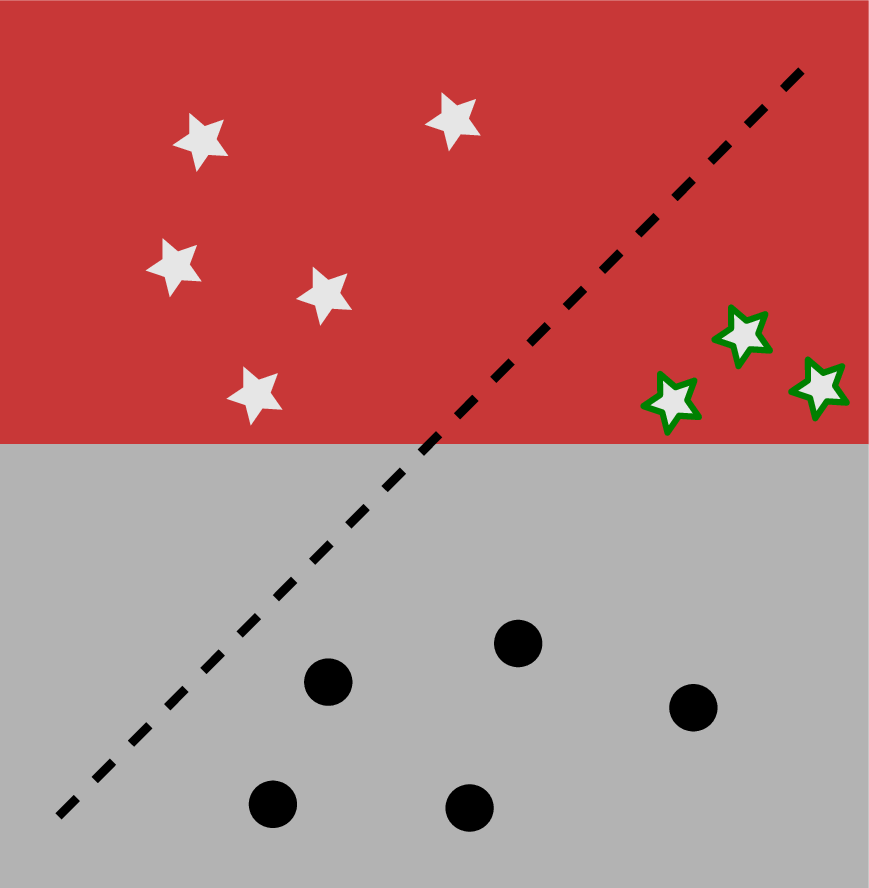}
			\caption{}
			\label{fig:worstd}
		\end{subfigure}

		\caption{Improved data augmentation}
	\end{figure}

	\section{Learning Sub-invariants} \label{sec:subinv}

	Next, we instantiate \tool for our second problem: learning
	sub-invariants.
	Given a program $\pprog = \while{G}{P} $ and
	a pair of pre- and post- expectations $(\pre, \post)$, we want to find a
	expectation $I$ such that $\pre \leq I$, and
	\[
	I \leq \Phi_{\post}^{\weakp} (I) := [\neg G] \cdot \post + [G] \cdot \weakp(P, I)
	\]
	Intuitively, $\Phi_{\post}^{\weakp}(I)$ computes the expected value of the
	expectation $I$ after one iteration of the loop.  We want to train a model $M$
	such that $M$ translates to an expectation $I$ whose expected value increases
	each iteration, and $\pre \leq I$.

	The high-level plan is the same as for learning exact invariants:
	we train a model to minimize a loss defined to capture
	the sub-invariant requirements.  We generate features $\mathcal{F}$ and sample
	initializations $\states$ as before. Then, from each $s \in \states$, we
	repeatedly run just the loop body $P$ and record the set of output states in
	$data$; this departs from our method for exact invariants, which repeatedly
	runs the entire loop to completion. Given this trace data, for any program
	state $s \in \states$ and expectation $I$, we can compute the empirical mean of
	$I$'s value after running the loop body $P$ on state $s$. Thus, we can
	approximate $\weakp(P, I)(s)$ for $s \in \states$ and use this estimate to
	approximate $\Phi_{\post}^{\weakp}(I)(s)$.  We then define a loss to sum up the
	violation of $I \leq \Phi_{\post}^{\weakp}(I)$ and $\pre \leq I$ on state $s \in
	\states$, estimated based on the collected data.

	The main challenge for our approach is that existing model tree learning
	algorithms do not support our loss function. Roughly speaking, model tree
	learners typically assume a node's two child subtrees can be learned separately;
	this is the case when optimizing on the loss we used for exact invariants, but
	this is \emph{not} the case for the loss for sub-invariants.

	To solve this challenge, we first broaden the class of models to neural
	networks.  To produce sub-invariants that can be verified, we still want to learn
	simple classes of models, such as piecewise functions of numerical expressions.
	Accordingly, we work with a class of neural architectures that can be translated
	into model trees, \emph{neural model trees}, adapted from neural decision trees
	developed by~\citet{DBLP:journals/corr/abs-1806-06988}. We defer the technical
	details of neural model trees to~\cref{subsec:dndtconstruction};
	for now, we can
	treat them as differentiable approximations of standard model trees; since they
	are differentiable they can be learned with gradient descent, which can optimize on
	the sub-invariant loss function.

	\paragraph{Outline.}
	We will discuss changes in \sample, \learninv and \verifyinv
	for learning sub-invariants but omit descriptions of
	\getfeatures, \getstates, \extractinv because
	\tool generates features, samples initial states
	and extracts expectations in the same way as in~\cref{sec:exactinv}.
	To simplify the exposition, we will assume \getfeatures generates the same set
	of features  $\mathcal{F} = \mathcal{F}_l = \mathcal{F}_m$ for
	model trees with linear models and model trees with multiplication models.

	\subsection{Sample Training Data (\sample)}
	Unlike when sampling data for learning exact invariants, here, \sample runs only
	one iteration of the given program $\pprog = \while{G}{P}$, that is, just $P$,
	instead of running the whole loop. Intuitively, this difference in data
	collection is because we aim to directly handle the sub-invariant condition,
	which encodes a single iteration of the loop. For exact invariants, our approach
	proceeded indirectly by learning the expected value of $\post$ after running the
	loop to termination.

	From any initialization $s_i \in \states$ such that $G$ holds on $s_i$,
	$\sample$ runs the loop body $P$ for $\nruns$ trials, each time restarting from
	$s_i$, and records the set of output states reached. If executing $P$ from $s_i$
	leads to output states $\{s_{i1}, \dots, s_{i\nruns}\}$, then \sample produces
	the training example:
	\begin{align*}
		(s_i, S_i)= \left(s_i, \left\{s_{i1}, \dots, s_{i\nruns} \right\} \right),
	\end{align*}
	For initialization $s_i \in \states$ such that $G$ is false on $s_i$,
	\sample simply produces $(s_i, S_i)= \left(s_i, \emptyset\right)$ since the loop
	body is not executed.

	\subsection{Learning a Neural Model Tree (\learninv)}
	\label{sec:sub:learninv}
	Given the dataset $data = \{ (s_1, S_1), \dots, (s_{K}, S_{K}) \}$ (with $K =
	\nstates$),
	we want to learn an expectation $I$ such that $\pre \leq I$
	and $I \leq \Phi^{\weakp}_{\post}(I)$.
	By case analysis on the guard $G$,
	the requirement $I \leq \Phi^{\weakp}_{\post}(I)$ can be split into two
	constraints:
	\[
	[G] \cdot I \leq [G] \cdot \weakp(P, I)
	\qquad\text{and}\qquad
	[\neg G] \cdot I \leq [\neg G] \cdot \post .
	\]
	If $I = \post + [G] \cdot I' $, then the second requirement
	reduces to $[\neg G] \cdot post E \leq [\neg G] \cdot post E$
	and is always satisfied.
	So to simplify the loss and training process,
	we again aim to learn an expectation $I$ of the
	form of $\post + [G] \cdot I'$.
	Thus, we want to train a model tree $T$ such that $T$
	translates into an expectation $I'$, and
	\begin{align}
		\pre &\leq \post + [G] \cdot I' \label{eq:sub:pre} \\
		[G] \cdot (\post + [G] \cdot I') 					&\leq [G] \cdot \weakp(P, \post + [G] \cdot I') \label{eq:sub:G}
	\end{align}
	Then, we define the loss of model tree $T$ on $data$ to be
	\begin{align*}
		err(T, data) \defeq err_1(T, data) +  err_2(T, data),
	\end{align*}
	where $err_1(T, data)$ captures~\cref{eq:sub:pre} and
	$err_2(T, data)$ captures~\cref{eq:sub:G}.

	Defining $err_1$ is relatively simple: we sum up the one-sided difference
	between $\pre(s)$ and $\post(s) + G(s) \cdot T(\mathcal{F}(s))$ across $s \in
	\mathit{states}$, where $T$ is the model tree getting trained and
	$\mathcal{F}(s)$ is the feature vector $\mathcal{F}$ evaluated on $s$.  That is,
	\begin{align}
		err_1 (T, data) := \sum_{i = 1}^{K} \max \left(0, \pre(s_i) - \post(s_i) - G(s_i) \cdot T(\mathcal{F}(s_i)) \right).
		\label{eq:err1}
	\end{align}
	Above, $\pre(s_i)$, $\post(s_i)$, and $G(s_i)$ are the value of expectations
	$\pre$, $\post$, and $G$ evaluated on program state $s_i$.

	The term $err_2$ is more involved.
	Similar to $err_1$, we aim to sum up the one-sided difference between
	two sides of~\cref{eq:sub:G} across state $s \in \states$.
	On program state $s$ that does not satisfy $G$, both sides are 0;
	for $s$ that satisfies $G$, we want to evaluate $\weakp(P, \post + [G] \cdot I')$
	on $s$, but we do not have exact access to $\weakp(P, \post + [G] \cdot I')$
	and need to approximate its value on $s$ based on sampled program traces.
	Recall that $\weakp(P, I)(s)$ is the
	expected value of $I$ after running program $P$ from $s$,
	and our dataset contains training examples $(s_i, S_i)$
	where $S_i$ is a set of states reached after running $P$ on an initial state
	$s_i$ satisfying $G$.
	Thus, we can approximate $[G] \cdot \weakp(P, \post + G \cdot I')(s_i)$ by
	\[
	G(s_i) \cdot \frac{1}{|S_i|} \cdot \sum_{s \in S_i} \left( \post(s) + G(s) \cdot I'(s) \right).
	\]
	To avoid division by zero when $s_i$ does not satisfy $G$ and $S_i$ is empty,
	we evaluate the expression in a short-circuit manner such that when $G(s_i) = 0$,
	the whole expression is immediately evaluated to zero.

	Therefore, we define
	\begin{align*}
		err_2 (T, data) = \sum_{i = 1}^{K}  \max\bigg( &0, G(s_{i}) \cdot \post(s_{i}) + G(s_{i}) \cdot  T(\mathcal{F}(s_{i}))   \\
		& -  G(s_{i}) \cdot \frac{1}{|S_i|} \cdot \sum_{s \in S_i} \big(\post(s) + G(s) \cdot T(\mathcal{F}(s)) \big) \bigg).
	\end{align*}
	Standard model tree learning algorithms do not support this kind of loss
	function, and since our overall loss $err(T, data)$ is the sum of $err_1(T,
	data)$ and $err_2(T, data)$, we cannot use standard model tree learning
	algorithm to optimize $err(T, data)$ either. Fortunately, gradient descent does
	support this loss function. While gradient descent cannot directly learn model
	trees (See ~\cref{subsec:probdndt}), we can use gradient descent to train a \emph{neural} model tree $T$ to minimize $err(T, data)$. The learned neural networks can be converted to model trees, and then converted to expectations as before (See ~\cref{subsec:dndtconstruction}).

	\subsubsection{Difficulty of Training with Standard Algorithms.}
	\label{subsec:probdndt}
	When calculating the error contributed by one training example in $err_l$,
	$err_m$ or $err_1$, the training algorithm only needs to evaluate the model
	tree	$T$ on the feature vector of \emph{one} program state:
	the data entry $(s_i, v_i)$ in $\mathit{data}$ contributes
	\[\left(\post(s_i) + G(s_i) \cdot T(\mathcal{F}_l(s_i)) - v_i \right)^2 \]
	to $err_l^2(T, \mathit{data})$ and similarly
	$\left(\post(s_i) + G(s_i) \cdot T(\mathcal{F}_m(s_i)) - v_i \right)^2 $
	to $err_m^2(T, \mathit{data})$;
	the data entry $(s_i, S_i)$ contributes
	\[\max \left(0, \pre(s_i) - \post(s_i) - G(s_i) \cdot T(\mathcal{F}(s_i)) \right)\]
    to $err_1^2(T, \mathit{data})$.
	For $f = \mathcal{F}_l(s_i), \mathcal{F}_m(s_i)$ or $\mathcal{F}(s_i)$,
	when we calculate $T(f)$ for a single $f$ as in above expressions,
	either $T$'s root is a leaf and we simply apply its leaf model to the $f$,
	or we pass $f$ to exactly \emph{one} children subtree $T'$ of $T$
	and recursively calculate $T'(f)$.
	We associate a training example $(s_i, v_i)$ (or $(s_i, S_i)$) to a subtree
	$T'$ if the feature vector $f$ of $s_i$ gets passed to $T'$.
	Thus, different subtrees get associated with disjoint sets of training
	examples, and the standard training algorithm for model trees
	is able to adopt the divide-and-conquer strategy and optimize children of $T$
	independently.

	However, the $err_2$ error of a model tree $T$ on
	a training example $(s_i, S_i)$ is
	\[
		\max\bigg( 0, G(s_{i}) \cdot \post(s_{i}) + G(s_{i}) \cdot  T(\mathcal{F}(s_{i}))
		 -  G(s_{i}) \cdot \frac{1}{|S_i|} \cdot \sum_{s \in S_i} \big(\post(s) + G(s) \cdot T(\mathcal{F}(s)) \big),
		\]
	which depends on the evaluated values of $T(\mathcal{F}(s_i))$ and
	$T(\mathcal{F}(s))$ for all $s \in S_i$.  Evaluating $T(\mathcal{F}(s_i))$ and
	all $T(\mathcal{F}(s))$ can use multiple children subtrees of $T$. Thus, we
	cannot associate one training example $(s_i, S_i)$ to exactly one
	children of $T$ and optimize children of $T$ independently.
	While children of $T$ still gets associated with disjoint sets of feature vectors,
	we cannot train children of $T$ to minimize $err_2$ just with
	sets of program states $s_i$ because unlike when learning exact invariants,
	now we do not know what $T(\mathcal{F}(s_i))$ should be without calculating
	$\sum_{s \in S_i} \big(\post(s) + G(s) \cdot T(\mathcal{F}(s))$.
	Furthermore, because of the use of $\max(0,
	-)$ function in $err_2$, we cannot solve the problem
	by rearranging the terms across training examples.

	\subsubsection{Constructions of Neural Model Trees}
	\label{subsec:dndtconstruction}
	To address the problem of optimizing on $err_2$, we consider another general
	training algorithm, \emph{gradient descent}, i.e., iteratively
	taking the gradient of the loss with respect to the trainable parameters and
	adjust the parameters along  the gradient to minimize the loss. Although
	gradient descent only provides theoretical guarantee of finding global minimum
	when the loss is convex, gradient descent and its stochastic variant have
	showed good performance across a wide range
	of problems, as demonstrated by the recent success of neural networks.
	To apply gradient descent, however, we need the training model to
	be differentiable with respect to trainable parameters, and model trees and
	decision trees are not differentiable with respect to each predicate in the
	internal node.
	To address this problem, we use a differentiable approximation of
	model trees based on neural networks, which we call \emph{neural
	model trees}, and train them using standard gradient descent method.
	we start with a model called  \emph{neural decision tree} developed
	by~\citet{DBLP:journals/corr/abs-1806-06988}.
  As in standard decision tree learning, they consider internal nodes predicates
	and leaf labels as trainable parameters.
  Unfortunately, if we let the training parameter at
	an internal node range over predicates of the form $f \leq c$, it is unclear
  how to develop a differentiable loss function to optimize.
	Thus, deviating from standard decision tree, neural decision trees assume
	there is a predicate of the form $f_i \leq c_i$ for each feature $f_i$ that we
	can split on and regard the cut point $c_i$ as the training parameter,
	and in addition, they use a smooth approximation of the predicate $f_i \leq c_i$
	so the loss becomes differentiable.

	Concretely, a neural decision tree is a function
	\[
		\nndt: \text{Trainable parameters } \Theta \times \text{Data} \to  \text{Labels}
\]
	where the trainable parameters includes cut points $c_i$ for each feature $f_i \in
	\mathcal{F}$ and numerical labels.
  The neural architecture developed by \citet{DBLP:journals/corr/abs-1806-06988}
  implements the map $\nndt$ in two stages: $\nndt = \labeling \circ \classify$,
  where
  \[
    \labeling: \text{Trainable parameters } \Theta \times \text{Data} \to \text{OneHot}(|2^{\mathcal{F}}|)
  \]
  takes input $(\{c_i\}, d)$ to a one-hot vector of length $2^{\mathcal{F}}$
  where the hot bit represents the leaf that the data example $d$ gets
  classified into according to the smoothed version of predicates $f_i \leq
  c_i$, and
  \[
    \classify: \text{OneHot}(|2^{\mathcal{F}}|)\to \text{Labels}
  \]
  assigns a label to each leaf.

  To approximate model trees with differentiable neural networks,
	we define
	\[
		\nnmt: \text{Trainable parameters } \Theta \times \text{Data} \to \mathbb{R}
	\]
	by having $\nnmt(\theta, d) = \regress(\classify (\theta, d)) (d)$
	where
		\[
			\regress:  \text{OneHot}(|2^{\mathcal{F}}|)\to \text{Leaf Models}
		\]
	associates each one-hot vector with a trainable leaf model.

	Our tool assumes the leaf models are linear models or multiplication models.
	When the leaf models are linear, each leaf model can be represented by a vector
	of linear coefficients, and we can represent $\regress$ as a $|2^{\mathcal{F}}| \times
	|\mathcal{F}|$ matrix. When we want to fit neural model
	trees with multiplication leaf models, we take the logarithm of data passed to
	the neural model, train a model with linear leaf models, and exponentiate the
	output. In both cases, $\nnmt$ is differentiable with respect to its trainable
	parameters, so we can apply standard stochastic gradient descent to train it.

	\subsection{Verify Extracted Expectations (\verifyinv)}
	The verifier \verifyinv is very similar to the one in~\cref{sec:exactinv}
	except here it solves a different optimization problem.
	For each candidate $inv$ in the given list $invs$, it looks for
	a set $S$ of program states such that $S$ includes
	\begin{align*}
		\mathbf{argmax}_s \pre(s) - inv(s) \qquad \text{ and }\qquad
		\mathbf{argmax}_s  G(s) \cdot I(s) - [G] \cdot \weakp(P,I)(s).
	\end{align*}
	As in our approach for exact invariant learning, the verifier aims to find
	counterexample states $s$ that violate at least one of these constraints by as
	large of a margin as possible; these high-quality counterexamples guide data
	collection in the following iteration of the CEGIS loop. Concretely, the
	verifier accepts $inv$ if it cannot find any program state $s$ where $\pre(s) -
	inv(s)$ or $G(s) \cdot I(s) - [G] \cdot \weakp(P,I)(s)$ is positive. Otherwise,
	it adds all states $s \in S$ with strictly positive margin to the set of
	counterexamples $cex$.

	\section{Evaluation} \label{sec:eval}

	We implemented our prototype in Python, using sklearn and tensorflow to fit
	model trees and neural model trees, and Wolfram Alpha to verify and perform
	counterexample generation. We have evaluated our tool on a set of 18 benchmarks
	drawn from different sources in prior work \cite{DBLP:conf/qest/GretzKM13,
		DBLP:conf/cav/ChenHWZ15, kaminski2017weakest}.
    \fix{All our benchmarks are almost surely terminating:
    for all benchmarks except \gambler, their almost sure terminations are witnessed by simple ranking super-martingales that are linear on variables and tests on variables, e.g., $[x \geq 0]$; the gambler's ruin problem is also well-studied to be AST \citep{Chatterjee2016, gamblernotes}. While we check this condition by hand, existing work has explored synthesis of such ranking super-martingale~(e.g., \citep{Chatterjee2016,Chatterjee:2016:AAQ:2837614.2837639,mciver2016new, huang2019modular, majumdar2024sound}). }

    Our experiments were designed to
	address the following research questions:
	\begin{description}
		\item[R1.] Can \tool synthesize exact invariants for a variety of programs?
		\item[R2.] Can \tool synthesize sub-invariants for a variety of programs?
	\end{description}

	We summarize our findings as follows:
	\begin{itemize}
		\item {\tool} successfully synthesized and verified exact invariants for 14/18
		benchmarks within a timeout of 300 seconds. Our tool was able to generate
		these 14 invariants in reasonable time, taking between 1 to 237 seconds. The
		sampling phase dominates the time in most cases.  We also compare {\tool} with
		a tool from prior literature, \textsc{Mora}~\cite{bartocci2020mora}. We found
		that \textsc{Mora} can only handle a restrictive set of programs and cannot
		handle many of our benchmarks.  We also discuss how our work compares with a
		few others in (\Cref{sec:rw}).
		\item To evaluate sub-invariant learning, we created multiple problem
		instances for each benchmark by supplying different pre-expectations. On a
		total of 34 such problem instances, {\tool} was able to infer correct
		invariants in 27 cases, taking between 7 to 102 seconds.
	\end{itemize}
	\iffull
	We present in the Appendix~\cref{app:results}
	\else
	We present in \href{https://arxiv.org/abs/2106.05421}{the extended version}
	\fi
	the tables of complete experimental results.  Because the training data we
	collect are inherently stochastic, the results produced by our tool are not
	deterministic.\footnote{%
		The code and data sampled in the trial that produced the tables in this paper
		can be found at \url{https://github.com/JialuJialu/Exist}.}
	As expected, sometimes different trials on the same benchmarks generate
	different sub-invariants; while the exact invariant for each benchmark is unique,
	\tool may also generate semantically equivalent but syntactically different
	expectations in different trials (e.g. it happens for \biasdir).

	\paragraph{Implementation Details.}
	For input parameters to \tool, we use $\nruns = 500$ and $\nstates = 500$.
	Besides input parameters listed in~\cref{fig:cegl}, we allow the user to supply a
	list of features as an optional input. In feature generation, \getfeatures
	enumerates expressions made up by program variables and user-supplied features
	\iffull
	according to a grammar described in~\cref{app:params}.
	\else
	according to a grammar.
	\fi
	Also, when incorporating counterexamples $cex$, we make 30 copies of each
	counterexample to give them more weights in the training.  All experiments were
	conducted on a MacBook Pro 2020 with M1 chip running macOS Monterey Version
	12.1.

	\subsection{R1: Evaluation of the Exact Invariant Method}
	\paragraph{Efficacy of Invariant Inference.}
	\begin{table}[t]
		\centering
		\caption{Exact Invariants generated by \tool}
		\label{table:exactinv}
		\begin{tabular}{ccccccc}
			\toprule
			\rowcolor{gray} \textbf{Name} &
			{$\post$} &
			\textbf{Learned Invariant} &
			\textbf{\Centerstack{ ST }} &
			\textbf{\Centerstack{ LT }} &
			\textbf{\Centerstack{ VT }} &
			\textbf{\Centerstack{ TT }} \\

			\midrule
			\addlinespace[0.8em]
			\biasdir &
			$x$ &
			\makecell{$x + [x == y] \cdot(-0.2\cdot x^2 - 0.2 \cdot y^2 $ \\
				$ -0.2 \cdot x \cdot y - 0.2 \cdot x - 0.2 \cdot y +0.5)$} &
			23.97 &
			0.34 &
			0.25 &
			24.56 \\
			\addlinespace[0.8em]
			\binzero &
			$x$ &
			\makecell{$x + [n > 0] \cdot ([y \neq 0] \cdot p \cdot n \cdot y ) $} &
			72.80 &
			5.09 &
			1.15 &
			79.04 \\
			\addlinespace[0.8em]
			\binone &
			$n$ &
			\makecell{$x + [n < M] \cdot (M \cdot p - n \cdot p) $} &
			25.67 &
			12.03 &
			0.22 &
			37.91 \\
			\addlinespace[0.8em]
			\bintwo &
			$x$ &
			\makecell{$x + [n > 0] \cdot(p \cdot n \cdot n $\\
				$ -p \cdot n \cdot y + n \cdot y)$} &
			84.64 &
			26.54 &
			0.48 &
			111.66 \\
			\addlinespace[0.8em]
			\deprv &
			$x\cdot y $ &
			- &
			- &
			- &
			- &
			- \\
			\addlinespace[0.8em]
			\detm &
			$count$ &
			\makecell{$count +  [x \leq 10] \cdot (11 -x) $} &
			0.09 &
			0.72 &
			0.06 &
			0.87 \\
			\addlinespace[0.8em]
			\duel &
			$t$ &
			- &
			- &
			- &
			- &
			- \\
			\addlinespace[0.8em]
			\fair &
			$count$ &
			\makecell{$count + [c1 + c2 == 0] \cdot $ \\
				$(p1+p2) / (p1+p2-p1 \cdot p2)$} &
			5.78 &
			1.62 &
			0.30 &
			7.69 \\
			\addlinespace[0.8em]
			\gambler &
			$z$ &
			\makecell{$z + [x > 0 \text{ and } y > x] \cdot $\\
				$ x \cdot (y-x)$} &
			112.02 &
			3.52 &
			9.97 &
			125.51 \\
			\addlinespace[0.8em]
			\geozero &
			$z$ &
			\makecell{$ z + [flip == 0] \cdot (1-p_1)/ p_1 $} &
			12.01 &
			0.85 &
			2.65 &
			15.51 \\
			\addlinespace[0.8em]
			\geoone &
			$z$ &
			\makecell{$ z + [flip == 0] \cdot (1-p_1)/ p_1 $} &
			20.30 &
			5.20 &
			3.57 &
			29.09 \\
			\addlinespace[0.8em]
			\geotwo &
			$z$ &
			\makecell{$ z + [flip == 0] \cdot (1-p_1)/ p_1 $} &
			10.78 &
			2.17 &
			0.12 &
			13.07 \\
			\addlinespace[0.8em]
			\geoar &
			$x$ &
			- &
			- &
			- &
			- &
			- \\
			\addlinespace[0.8em]
			\linexp &
			$z$ &
			- &
			- &
			- &
			- &
			- \\
			\addlinespace[0.8em]
			\mart &
			$rounds$&
			\makecell{$rounds + [b > 0] \cdot (1/p) $} &
			24.10 &
			3.83 &
			0.05 &
			27.98 \\
			\addlinespace[0.8em]
			\prinsys &
			$[x==1]$ &
			\makecell{$ [x == 1] + [x == 0] \cdot (1-p_2)$} &
			1.60 &
			0.17 &
			1.25 &
			3.02 \\
			\addlinespace[0.8em]
			\revbin &
			$z$ &
			$z + [x > 0] \cdot (x/p)$ &
			234.64 &
			3.13 &
			0.14 &
			237.92 \\
			\addlinespace[0.8em]
			\progsum &
			$x$ &
			\makecell{$x + [n > 0] \cdot (0.5 \cdot p \cdot n^2 + $ \\
				$ 0.5 \cdot p \cdot n)$} &
			102.12 &
			34.61 &
			26.74 &
			163.48 \\
			\addlinespace[0.8em]
		\end{tabular}
	\end{table}

	{\tool} was able to infer exact invariants in 14/18 benchmarks. Out
  of 14 \fix{successfully inferred} benchmarks, only 2 of them need user-supplied features
($n \cdot p$ for \bintwo and \progsum).
	~\Cref{table:exactinv}
	shows the postexpectation ($\post$), the inferred
	invariant (Invariant), sampling time (ST), learning time (LT),
	verification time (VT) and the total time (TT) for a few benchmarks.
	For generating exact invariants, the running time of {\tool} is dominated by the
	sampling time. However, this phase can be parallelized easily.

  \fix{
    We (manually) check that all the inferred invariants only evaluates to non-negative values if all the program variables take non-negative values, which is required for them to be expectations as defined in~\cref{def:expectation} and to apply~\cref{theorem:equivalence}. We then (manually) check whether the inferred invariants are provably the weakest preexpectations according to~\cref{theorem:equivalence}.
    We find 13 out of the 14 exact invariants inferred by \tool satisfy at least one condition, and thus, are provably the weakest preexpectation:
    the invariants inferred for \binzero, \binone, \bintwo and \progsum satisfy condition (a), and the subinvariants inferred for the rest
  of the benchmarks except \gambler all satisfy the condition (b).
The invariant \tool inferred for \gambler, $z + [x > 0 \text{ and } y > x] \cdot x \cdot (y-x)$, does not satisfy any of the conditions, but it is sound according to known results about random walks: the postexpectation $z$ increases 1 at each iteration, and recurrence analysis shows that $x \cdot (y - x)$ is the expected number of iterations \gambler~\citep{gamblernotes}. }

	\paragraph{Failure Analysis. }
	\tool failed to generate invariants for 4/18 benchmarks. For two of them,
	\tool was able to generate expectations that are very close to an
	invariant (\deprv and \linexp);
	for the third failing benchmarks (\duel),
	the ground truth invariant is very complicated.
	For \linexp, while a
	correct invariant is $z + [n > 0] \cdot 2.625 \cdot n$, \tool generates
	expectations like $z + [n > 0] \cdot (2.63\cdot n-0.02)$ as candidates. For
	\deprv, a correct invariant is $x\cdot y + [n>0] \cdot (0.25 \cdot n^2 +0.5
	\cdot n \cdot x + 0.5\cdot n \cdot y -0.25 \cdot n)$, and in our experiment
	\tool generates $0.25 \cdot n^2 +0.5 \cdot n \cdot x + 0.5\cdot n \cdot y -0.27
	\cdot n-0.01 \cdot x+0.12$. In both cases, the ground truth invariants use
	coefficients with several digits, and since learning from data is inherently
	stochastic, \tool cannot generate them consistently. In our experiments, we
	observe that our CEGIS loop does guide the learner to move closer to the
	correct invariant in general, but sometimes progress obtained in multiple iterations can
	be offset by noise in one iteration. For \geoar, we observe the verifier
	incorrectly accepted the complicated candidate invariants generated by
	the learner because Wolfram Alpha was not able to find valid counterexamples
	for our queries.

\paragraph*{Comparison with Previous Work.}

	\begin{table}[t]
		\centering
		\caption{Encoding of \geoar and \mart in \textsc{Mora} Syntax}
		\label{table:moraencoding}
		\begin{tabular}{|c|c|}
			\hline
			\rowcolor{gray} \textbf{Program} &
			\textbf{Encoding} \\ \hline\hline

			\begin{minipage}[t]{0.45\textwidth}
				\centering
				\begin{minted}[xleftmargin=\parindent, tabsize=2, numbersep=-1em, fontsize=\footnotesize, linenos, escapeinside=||, mathescape=true]{python}
					bool z, int x,y, float p
					while (z |$\neq$| 0) do
					y |$\gets$| y + 1;
					c |$\overset{\$}{\leftarrow}$| Bernoulli(p)
					if c then z |$\gets$| 0
					else x |$\gets$| x + y
				\end{minted}
			\end{minipage}

			&
			\begin{minipage}[t]{0.45\textwidth}
				\centering
				\begin{minted}[xleftmargin=\parindent, tabsize=2, numbersep=-1em, fontsize=\footnotesize, linenos, escapeinside=||, mathescape=true]{python}
					int z = 1, x = 0, y = 0
					while true:
					y |$\gets$| y + z;
					z |$\gets$| 0 |$@$| p;z
					x |$\gets$| x + y * z
				\end{minted}
			\end{minipage}\\ \hline

			\begin{minipage}[t]{0.45\textwidth}
				\centering
				\begin{minted}[xleftmargin=\parindent, tabsize=2, numbersep=-1em, fontsize=\footnotesize, linenos, escapeinside=||, mathescape=true]{python}
					int c, b, rounds, float p
					while (b |$>$| 0) do
					d |$\overset{\$}{\leftarrow}$| Bernoulli(p)
					if d then
					c |$\gets$| c + b;
					b |$\gets$| 0;
					else
					c |$\gets$| c - b;
					b |$\gets$| 2 * b;
					rounds |$\gets$| rounds + 1;
				\end{minted}
			\end{minipage}

			&
			\begin{minipage}[t]{0.45\textwidth}
				\centering
				\begin{minted}[xleftmargin=\parindent, tabsize=2, numbersep=-1em, fontsize=\footnotesize, linenos, escapeinside=||, mathescape=true]{python}
					int c = 0
					b = 1
					d = 0
					rounds = 1
					while true:
					d |$\gets$| 1 |$@$| p;d
					c |$\gets$| c + b*(d-1)+b*d
					b |$\gets$| 2*b*(1-d)
					rounds |$\gets$| rounds + 1 - d
				\end{minted}
			\end{minipage}\\ \hline

		\end{tabular}
	\end{table}

	There are few existing tools that can automatically compute expected values
	after probabilistic loops. We experimented with one such tool, called
	\textsc{Mora}~\citep{bartocci2020mora}.
	We managed to encode our benchmarks \geozero, \binzero,
	\bintwo, \geoone, \geoar, and \mart in their syntax.  Among them, \textsc{Mora}
	fails to infer an invariant for \geoone, \geoar, and \mart.  We also tried to
	encode our benchmarks \fair, \gambler, \binone, and \revbin but found
	\textsc{Mora}'s syntax was too restrictive to encode them. \Cref{table:moraencoding} shows how we encoded two of our benchmarks into \textsc{Mora}.

	\subsection{R2: Evaluation of the Sub-invariant Method}
	\begin{table}[h]
		\centering
		\caption{Sub-Invariants generated by \tool}
		\label{table:subinv1}
		\begin{tabular}{cccccccc}
			\toprule
			\rowcolor{gray} \textbf{Name} &
			{$\post$} &
			$\pre$ &
			\textbf{Learned Invariant} &
			\textbf{\Centerstack{ ST }} &
			\textbf{\Centerstack{ LT }} &
			\textbf{\Centerstack{ VT }} &
			\textbf{\Centerstack{ TT }} \\
			\midrule
			\addlinespace[0.8em]

			\multirow{2}{*}{\biasdir} &
			\multirow{2}{*}{$x$} &
			$[x \neq y]\cdot x$ &
			\makecell{$x + [x == y]\cdot $\\$(0.1 \cdot p - 0.5 \cdot x$\\$-0.5 \cdot y + 0.1)$} &
			12.37 &
			21.36 &
			0.75 &
			34.48 \\
			\addlinespace[0.8em]

			&
			&
			$[x == y]\cdot 1/2$ &
			\makecell{$x + [x == y]\cdot $\\$ (-0.5 \cdot x  $ \\ $ -0.5 \cdot y + 0.5)$} &
			12.33&
			27.12&
			0.10&
			39.55\\
			\addlinespace[0.8em]

			\multirow{2}{*}{\binzero} &
			\multirow{2}{*}{$x$} &
			\makecell{$x+[n > 0] \cdot $ \\ $ (p \cdot n \cdot y) $} &
			- &
			- &
			- &
			- & - \\
			\addlinespace[0.8em]

			&
			&
			$x$ &
			$x$ &
			16.69 &
			21.47 &
			0.44 &
			38.59 \\
			\addlinespace[0.8em]

			\multirow{2}{*}{\binone} &
			\multirow{2}{*}{$n$} &
			\makecell{$x + [n < M] \cdot $ \\ $ (p \cdot M - p \cdot n)$} &
			- &
			- &
			- &
			- &
			- \\
			\addlinespace[0.8em]

			&
			&
			$n$ &
			$x $ &
			8.55 &
			17.84 &
			0.10 &
			26.52 \\
			\addlinespace[0.8em]

			\multirow{2}{*}{\bintwo} &
			\multirow{2}{*}{$x$} &
			\makecell{$x + [n > 0] \cdot $ \\ $ (1-p) \cdot n\cdot y$} &
			- &
			- &
			- &
			- &
			- \\
			\addlinespace[0.8em]

			&
			&
			$x$ &
			$x$ &
			16.71 &
			21.64 &
			0.78 &
			39.13 \\
			\addlinespace[0.8em]

			\multirow{2}{*}{\deprv} &
			\multirow{2}{*}{$x\cdot y$} &
			\makecell{$x \cdot y + [n>0]\cdot $ \\ $ (1/4\cdot n\cdot n )$} &
			- &
			- &
			- &
			- &
			- \\
			\addlinespace[0.8em]

			&
			&
			$x \cdot y$ &
			$x \cdot y$ &
			16.08 &
			17.04 &
			0.17 &
			33.28 \\
			\addlinespace[0.8em]

			\multirow{2}{*}{\detm} &
			\multirow{2}{*}{$count$} &
			\makecell{$count +$ \\  $[x <= 10]\cdot 1$} &
			\makecell{$count +$ \\ $[x <= 10]\cdot 1$} &
			3.95 &
			17.67&
			0.03&
			21.66\\
			\addlinespace[0.8em]

			&
			&
			$count$ &
			$count$ &
			5.99 &
			10.31 &
			0.04 &
			16.35 \\
			\addlinespace[0.8em]

			\duel &
			$t$ &
			\makecell{$c\cdot(-p_1 + p_2 $ \\ $- p_1 \cdot p_2)$ \\
				$+1$} &
			- &
			- &
			- &
			- &
			- \\
			\addlinespace[0.8em]

			\multirow{2}{*}{\fair} &
			\multirow{2}{*}{$count$} &
			\makecell{$count + $ \\ $ [c_1 + c_2 == 0] \cdot $ \\ $ (p_1 + p_2)$} &
			\makecell{$[ c_1 + c_2 == 0]\cdot $ \\ $(p_1 + p_2) $\\$+count$} &
			8.74 &
			25.84 &
			0.27 &
			34.85 \\
			\addlinespace[0.8em]

			&
			&
			$count$&
			\makecell{$count$} &
			6.71 &
			11.73 &
			0.40 &
			18.85 \\
			\addlinespace[0.8em]

			\multirow{2}{*}{\mart} &
			\multirow{2}{*}{$rounds$} &
			\makecell{$rounds + $ \\ $[b > 0] \cdot 1$} &
			\makecell{$rounds + $ \\ $[b > 0] \cdot 1$} &
			17.68 &
			31.49 &
			0.11 &
			49.27 \\
			\addlinespace[0.8em]

			&
			&
			$rounds$ &
			\makecell{$rounds$} &
			16.61 &
			21.32 &
			0.16 &
			38.09 \\
			\addlinespace[0.8em]

			\bottomrule
		\end{tabular}
	\end{table}

	\begin{table}[h]
		\raggedright
		\caption{Table \ref{table:subinv1} Continued}
		\label{table:subinv2}
		\begin{tabular}{cccccccc}
			\toprule
			\rowcolor{gray} \textbf{Name} &
			{$\post$} &
			$\pre$ &
			\textbf{Learned Invariant} &
			\textbf{\Centerstack{ ST }} &
			\textbf{\Centerstack{ LT }} &
			\textbf{\Centerstack{ VT }} &
			\textbf{\Centerstack{ TT }} \\
			\midrule
			\addlinespace[0.8em]

			\multirow{2}{*}{\gambler} &
			\multirow{2}{*}{$z$} &
			\makecell{$z$} &
			\makecell{$z$} &
			6.99 &
			12.56 &
			0.43 &
			19.98 \\
			\addlinespace[0.8em]

			&
			&
			\makecell{$x \cdot (y-x)$} &
			\makecell{$z + $ \\ $[x > 0 \text{ \& } y > x] \cdot$ \\ $ x \cdot (y-x)$} &
			7.31 &
			28.87 &
			8.29 &
			44.46 \\
			\addlinespace[0.8em]

			\multirow{3}{*}{\geozero} &
			\multirow{3}{*}{$z$} &
			\makecell{$z+ $ \\ $[flip == 0] \cdot$ \\ $(1-p_1)$} &
			\makecell{$z+ $ \\ $[flip == 0] \cdot$ \\ $(1-p_1)$} &
			8.69 &
			28.04 &
			0.10 &
			36.84 \\
			\addlinespace[0.8em]

			&
			&
			$z$ &
			$z$ &
			8.08 &
			12.01 &
			3.62 &
			23.71 \\
			\addlinespace[0.8em]

			&
			&
			\makecell{$[flip == 0]\cdot $ \\ $ (1-p_1))$}&
			\makecell{$z+ $ \\ $[flip == 0] \cdot$ \\ $(1-p_1)$}&
			8.70 &
			26.13 &
			0.19 &
			35.02 \\
			\addlinespace[0.8em]

			\geoone&
			$z$ &
			$z$ &
			\makecell{z} &
			8.80 &
			13.66 &
			0.03 &
			22.48 \\
			\addlinespace[0.8em]

			\geotwo &
			$z$ &
			$z$ &
			$z$ &
			8.19 &
			14.49 &
			0.05 &
			22.73 \\
			\addlinespace[0.8em]

			\multirow{2}{*}{\geoar} &
			\multirow{2}{*}{$x$} &
			\makecell{$x+[z!=0] \cdot $ \\ $y \cdot(1-p)/p$} &
			- &
			- &
			- &
			- &
			- \\
			\addlinespace[0.8em]

			&
			&
			$x$ &
			$x$ &
			8.51 &
			40.98 &
			0.39 &
			49.89 \\
			\addlinespace[0.8em]

			\multirow{2}{*}{\linexp} &
			\multirow{2}{*}{$z$} &
			\makecell{$z + [n>0]\cdot $ \\ $2$} &
			\makecell{$[n > 0]\cdot $ \\ $(n + 1)$} &
			53.72 &
			30.01 &
			0.35 &
			84.98 \\
			\addlinespace[0.8em]

			&
			&
			\makecell{$z + [n>0]\cdot $ \\ $ 2 \cdot n$ }&
			\makecell{$z + [n>0]\cdot $ \\ $ 2 \cdot n$ } &
			29.18 &
			28.61 &
			0.68 &
			58.48 \\
			\addlinespace[0.8em]

			\prinsys &
			\makecell{$[x==1] \cdot  $ \\ $1$} &
			\makecell{$[x==1] \cdot  $ \\ $1$} &
			\makecell{$[x==1]$} &
			1.10 &
			5.85 &
			0.33 &
			7.29 \\
			\addlinespace[0.8em]

			\multirow{2}{*}{\revbin} &
			\multirow{2}{*}{$z$} &
			\makecell{$z + [x > 0] \cdot $ \\ $ x$ }&
			\makecell{$z + [x > 0]\cdot $ \\ $ x / p $ }&
			18.17 &
			71.15 &
			2.17 &
			91.55 \\
			\addlinespace[0.8em]

			&
			&
			$z$ &
			$z$ &
			15.62 &
			18.74 &
			0.06 &
			34.42 \\
			\addlinespace[0.8em]

			\multirow{2}{*}{\progsum} &
			\multirow{2}{*}{$x$} &
			\makecell{$x + [n > 0]\cdot $ \\ $ (p\cdot n\cdot n/2)$} &
			- &
			- &
			- &
			- &
			- \\
			\addlinespace[0.8em]

			&
			&
			\makecell{$x + [n > 0]\cdot $ \\ $ (p\cdot n /2)$} &
			\makecell{$x + [n>0]\cdot$ \\ $ (p \cdot n)$} &
			19.60&
			76.71 &
			5.94 &
			102.29
			\\
			\addlinespace[0.8em]

			\bottomrule
		\end{tabular}
	\end{table}

	\paragraph{Efficacy of invariant inference.}

    {\tool} is able to synthesize sub-invariants for 27/34
    benchmarks. Two out of 27 successfully inferred benchmarks use
    user-supplied features -- \gambler with pre-expectation $x \cdot (y-x)$
    uses $(y-x)$, and \progsum with pre-expectation $x + [x>0] \cdot (p \cdot
    n/2)$ uses $p \cdot n$. Contrary to the case for exact invariants, the
    learning time dominates. This is not surprising: the sampling time is
    shorter because we only run one iteration of the loop, but the learning
    time is longer as we are optimizing a more complicated loss function.

  \fix{
  We check whether the subinvariants synthesized by \tool satisfy one of the
  conditions (a), (b) or (c) in~\cref{theorem:equivalence}, and thus, provably
  lower bounds the weakest preexpectation. There are two subinvariants that
  do not satisfy any of the conditions:
  \gambler with inferred subinvariant $z + [x > 0 \text{ and } y > x] \cdot
  x \cdot (y - x)$ and \geoar with the inferred invariants $x$;  both of them do lower bound the expectations calculated manually through~\cref{psemantics}.
  For the rest of 25 subinvariants inferred by \tool,
  the subinvariants inferred by \tool for \binzero, \binone, \bintwo, \linexp, \deprv and
  \progsum satisfy condition (a), and the subinvariants inferred for the rest
  of the benchmarks satisfy the condition (b). }
  As before,~\cref{table:subinv1} reports the details for all these benchmarks.

	One interesting thing that we found when gathering benchmarks is that for many
	loops, pre-expectations used by prior work or natural choices of pre-expectations
	are themselves sub-invariants.  Thus, for some instances, the sub-invariants
	generated by \tool is the same as the pre-expectation $\pre$ given to it as
	input. However, \tool is not checking whether the given $\pre$ is a
	sub-invariant: the learner in \tool does not know about $\pre$ besides the value
	of $\pre$ evaluated on program states. Also, we also designed benchmarks where
	pre-expectations are \emph{not} sub-invariants (\biasdir with $\pre = [x \neq y]
	\cdot x$, \deprv with $\pre = x \cdot y + [n > 0] \cdot 1/4 \cdot n^2$, \gambler
	with $\pre = x \cdot (y-x)$, \geozero with $\pre = [flip == 0] \cdot (1-p1)$),
	and \tool is able to generate sub-invariants for 3/4 such benchmarks.

  \paragraph{Failure Analysis.}
	On program instances where \tool fails to generate a sub-invariant, we observe
	two common causes. First, gradient descent seems to get stuck in local minima
	because the learner returns suboptimal models with relatively low loss.  The
	loss we are training on is very complicated and likely to be highly non-convex,
	so this is not surprising.  Second, we observed inconsistent behavior due to
	noise in data collection and learning.  For instance, for \geoar with $\pre = x
	+ [z \neq 0] \cdot y \cdot (1-p)/p$, \tool could sometimes find a sub-invariant
	with supplied feature $(1-p)$, but we could not achieve this result
	consistently.

	\paragraph{Comparison with Learning Exact Invariants. }
	The performance of \tool on learning sub-invariants is less sensitive to the
	complexity of the ground truth invariants.  For example, \tool is not able to
	generate an exact invariant for \linexp as its exact invariant is complicated,
	but \tool is able to generate sub-invariants for \linexp.  However, we also
	observe that when learning sub-invariants, \tool returns complicated expectations
	with high loss more often.

		\section{Related Work} \label{sec:rw}

		\paragraph*{Invariant Generation for Probabilistic Programs.}
		There has been a steady line of work on probabilistic invariant generation
		over the last few years. The \textsc{Prinsys}
		system~\citep{DBLP:conf/qest/GretzKM13} employs a template-based approach to
		guide the search for probabilistic invariants. \textsc{Prinsys} is able encode
		invariants with guard expressions, but the system doesn't produce invariants
		directly---instead, \textsc{Prinsys} produces logical formulas encoding the
		invariant conditions, which must be solved manually.

    \citet{DBLP:conf/cav/ChenHWZ15} proposed a counterexample-guided approach
    to find polynomial invariants, by applying Lagrange interpolation. However, invariants
		involving guard expressions---common in our examples---cannot be found,
		since they are not polynomials. Additionally,
		\citet{DBLP:conf/cav/ChenHWZ15} uses a weaker notion of invariant, which
		only needs to be correct on certain initial states; our tool generates
		invariants that are correct on all initial states. \citet{feng2017finding}
		improves on \citet{DBLP:conf/cav/ChenHWZ15} by using \emph{Stengle's
			Positivstellensatz} to encode invariants constraints as a semidefinite
		programming problem. Their method can find polynomial sub-invariants that are
		correct on all initial states. However, their approach cannot synthesize
		piecewise linear invariants, and their implementation has additional
		limitations and could not be run on our benchmarks.

    Subsequent to the original publication of our results,
    \citet{DBLP:conf/tacas/BatzCJKKM23} proposed a different method to
    synthesize invariant expectations. Their approach is based on a rich class
    of templates with numerical-valued holes, and uses an efficient CEGIS loop
    to improve the invariant expectations. Unlike our approach,
    \citet{DBLP:conf/tacas/BatzCJKKM23} use access to the program source code.
    In this way, they are able to use a verifier to check correctness of
    invariants and find counterexamples within their CEGIS loop. In contrast,
    our approach does not rely on direct verification during synthesis since our
    method does not have access to the internals of the program.

		There is also a line of work on abstract interpretation for analyzing
		probabilistic programs; \citet{DBLP:conf/sas/ChakarovS14} search for linear
		expectation invariants using a ``pre-expectation closed cone domain'', while
		recent work by \citet{DBLP:conf/pldi/WangHR18} employs a sophisticated
		algebraic program analysis approach.

		Another line of work applies \emph{martingales} to derive insights of
		probabilistic programs.  \citet{Chakarov-martingale} showed several
		applications of martingales in program analysis, and \citet{BEFFH16} gave a
		procedure to generate candidate martingales for a probabilistic program;
		however, this tool gives no control over which expected value is
		analyzed---the user can only guess initial expressions and the tool generates
		valid bounds, which may not be interesting. Our tool allows the user to pick
		which expected value they want to bound.

		Another line of work for automated reasoning uses \emph{moment-based
			analysis}. \citet{bartocci2019automatic,bartocci2020mora} develop the
		\textsc{Mora} tool, which can find the moments of variables as functions of
		the iteration for loops that run forever by using ideas from computational
		algebraic geometry and dynamical systems. This method is highly efficient
		and is guaranteed to compute moments exactly.  However, there are two
		limitations. First, the moments can give useful insights about the
		distribution of variables' values after each iteration, but they are
		fundamentally different from our notion of invariants which allow us to
		compute the expected value of any given expression \emph{after termination}
		of a loop. Second, there are important restrictions on the probabilistic
		programs. For instance, conditional statements are not allowed and the use
		of symbolic inputs is limited. As a result, most of our benchmarks cannot be
		handled by \textsc{Mora}.

		In a similar vein, ~\citet{kura2019tail,wang2021central} bound higher \emph{central moments}
		for running time and other monotonically increasing quantities. Like our
		work, these works consider probabilistic loops that terminate. However,
		unlike our work, they are limited to programs with constant size increments.

		\paragraph*{Data-driven Invariant Synthesis.}
		We are not aware of other data-driven methods for learning probabilistic
		invariants, but a recent work \citet{DBLP:conf/cav/AbateGR20} proves
		probabilistic termination by learning ranking supermartingales from trace
		data. Our method for learning sub-invariants (\cref{sec:subinv}) can be seen
		as a natural generalization of their approach. However, there are also
		important differences. First, we are able to learn general sub-invariants,
		not just ranking supermatingales for proving termination.  Second, our
		approach aims to learn model trees, which lead to simpler and more
		interpretable sub-invariants. In contrast, Abate, et
		al.~\citep{DBLP:conf/cav/AbateGR20} learn ranking functions encoded as
		two-layer neural networks.

		Data-driven inference of invariants for deterministic programs has drawn a lot
		of attention, starting from  \textsc{Daikon}~\cite{daikon}. ICE learning with
		decision trees~\cite{icedt} modifies the decision tree learning algorithm to
		capture \textit{implication counterexamples} to handle inductiveness.
		\textsc{Hanoi}~\cite{Miltner2020} uses counterexample-based inductive
		synthesis (CEGIS)~\cite{sketch} to build a data-driven invariant inference
		engine that alternates between weakening and strengthening candidates for
		synthesis. Recent work uses neural networks to learn invariants
		\cite{code2inv}. These systems perform classification, while our work
    uses regression. Data from fuzzing has been used for \textit{almost correct}
    inductive invariants for programs with closed-box
    operations~\cite{Lahiri:22}.

		\paragraph*{Probabilistic Reasoning with Pre-expectations.}
		Following Morgan and McIver, there are now pre-expectation calculi for
		domain-specific properties, like expected runtime~\citep{KaminskiKMO16} and
		probabilistic sensitivity~\citep{ABHKKM19}. All of these systems
		define the pre-expectation for loops as a least fixed-point, and practical
		reasoning about loops requires finding an invariant of some kind.

\section{Conclusion}
Inspired by data-driven invariant generation techniques for standard programs,
we present the first data-driven invariant generation algorithm for
probabilistic program.  Our method is the first to be able to learn exact
piecewise linear probabilistic invariants fully automatically, without relying
on templates or manually solving logical formulas.

Going forward, one potential
direction is to improve sampling performance. While our tool finds invariants in
a reasonable amount of time, it needs many input-output traces in order to
reliably find invariants. Methods from statistics, like boosting, might be
useful to increase stability of our learning approach.  More broadly, a natural
question is whether other quantitative invariants could be learned through
regression, rather than classification. In general, our work further strengthens the research direction invested in exploring the synergy between machine learning and formal methods is solving hard tasks like program verification.

		\subsection*{Acknowledgements}
    This work is in part supported by National Science Foundation grant
    \#1943130 and \#2152831. We also thank Google Research for supporting our research. We thank Ugo Dal Lago, I\c{s}il Dillig, IIT Kanpur PRAISE
    group, Cornell PL Discussion Group, and all reviewers for helpful feedback.
    We also thank Anmol Gupta in IITK for building a prototype verifier using
    Mathematica.

\subsection*{Data availability statement}
The source code of the implementation and data is publicly available at: \\ \url{https://github.com/JialuJialu/Exist}
		%
		%
		%

		\begin{appendices}

			\section{Benchmark Programs} \label{app:bench}
			For benchmarks that have multiple similar variations,
			we omit some simpler ones.

			\begin{figure*}
				\begin{lstlisting}
					int z, bool flip, float p1
					while (flip == 0):
						d = bernoulli.rvs(size=1, p=p1)[0]
						if d:
							flip = 1
						else:
							z = z + 1
				\end{lstlisting}
				\caption{\geozero: Adapted from Listing 1 from Prinsys paper}
				\label{fig:geo0}
			\end{figure*}

			\begin{figure*}
				\begin{lstlisting}
					int z, x, bool flip, float p1
					while (flip == 0):
						d = bernoulli.rvs(size=1, p=p1)[0]
						if d:
							flip = 1
						else:
							x = x * 2
							z = z + 1
				\end{lstlisting}
				\caption{\geoone: Adapted from~\cref{fig:geo0}}
				\label{fig:geo0a}
			\end{figure*}

			\begin{figure*}
				\begin{lstlisting}
					int count, bool c1, c2 float p1, p2
					while not (c1 or c2):
						c1 = bernoulli.rvs(size=1, p=p1)[0]
						if c1:
							count = count + 1
						c2 = bernoulli.rvs(size=1, p=p2)[0]
						if c2:
							count = count + 1\end{lstlisting}
				\caption{\fair: Two coins in one loop}
				\label{fig:ex1}
			\end{figure*}

			\begin{figure*}
				\begin{lstlisting}
					int c, b, rounds, float p
					while b > 0:
						d = bernoulli.rvs(size=1, p=p)
						if d:
							c = c+b
							b = 0
						else:
							c = c-b
							b = 2*b
						rounds += 1
				\end{lstlisting}
				\caption{\mart: Martingale, Listing 4 from Prinsys paper}
				\label{fig:ex2}
			\end{figure*}

			\begin{figure*}
				\begin{lstlisting}
					int x, y, z, float p
					while 0 < x and x < y:
						d = bernoulli.rvs(size=1, p=p)[0]
						if d:
							x = x + 1
						else:
							x = x - 1
							z = z + 1
						rounds += 1
					\end{lstlisting}
				\caption{\gambler: Gambler's ruin}
				\label{fig:ex4}
			\end{figure*}

			\begin{figure*}
				\begin{lstlisting}
					int x, y, z, float p
					while not (z == 0):
						y = y + 1
						d = bernoulli.rvs(size=1, p=p)[0]
						if(d):
							z = 0
						else:
							x = x + y
				\end{lstlisting}
				\caption{\geoar: Geometric distribution mixed with Arithmetic progression}
			\end{figure*}

			\begin{figure*}
				\begin{lstlisting}
					int x, y, n, float p
					while(n > 0):
						d = bernoulli.rvs(size=1, p=p)[0]
						if(d):
							x = x + y
						n = n-1
					\end{lstlisting}
				\caption{\binzero: A plain binomial distribution}
				\label{fig:ex7}
			\end{figure*}

			\begin{figure*}
				\begin{lstlisting}
					int x, n, M, float p
					while n - M < 0:
						d = bernoulli.rvs(size=1, p=p)[0]
						if d:
							x = x + 1
							n = n + 1
				\end{lstlisting}
				\caption{\binone: A binomial distribution with an unchanged variable}
				\label{fig:ex18}
			\end{figure*}

			\begin{figure*}
				\begin{lstlisting}
					int x, y, n, float p
					while(n > 0):
						d = bernoulli.rvs(size=1, p=p)[0]
						if(d):
							x = x + n
						else:
							x = x + y
							n = n-1
				\end{lstlisting}
				\caption{\bintwo: Binomial distribution mixed with arithmetic progression sums}
				\label{fig:ex8p}
			\end{figure*}

			\begin{figure*}
				\begin{lstlisting}
					int n, count
					while(n > 0):
						x1 = bernoulli.rvs(size=1, p=0.5)[0]
						x2 = bernoulli.rvs(size=1, p=0.5)[0]
						x3 = bernoulli.rvs(size=1, p=0.5)[0]
						n = n - 1
						c1 = x1 or x2 or x3
						c2 = (not x1) or x2 or x3
						c3 = x1 or (not x2) or x3
						count = count + c1 + c2 + c3
				\end{lstlisting}
				\caption{\linexp: Mix of binomial distribution and linearity of expectation}
				\label{fig:ex20}
			\end{figure*}

			\begin{figure*}
				\begin{lstlisting}
					int x, n, float p
					while(n > 0):
						d = bernoulli.rvs(size=1, p=p)[0]
						if(d):
							x = x + n
						n = n - 1
				\end{lstlisting}
				\caption{\progsum: Probabilistic sum of arithmetic series}
				\label{fig:ex9}
			\end{figure*}

			\begin{figure*}
				\begin{lstlisting}
					int x, y, n, float p
					while(n > 0):
						d = bernoulli.rvs(size=1, p=p)[0]
						if(d):
							x = x + 1
						else:
							y = y + 1
						n = n - 1
				\end{lstlisting}
				\caption{\deprv: Product of dependent random variables}
				\label{fig:ex10}
			\end{figure*}

			\begin{figure*}
				\begin{lstlisting}
					int x, float p1, p2
					while(x == 0):
						d1 = bernoulli.rvs(size=1, p=p1)[0]
						if(d1):
							x = 0
						else:
							d2 = bernoulli.rvs(size=1, p=p2)[0]
							if(d2):
								x = -1
							else:
								x = 1
				\end{lstlisting}
				\caption{\prinsys: Listing 2 from Prinsys paper}
				\label{fig:ex12}
			\end{figure*}

			\begin{figure*}
				\begin{lstlisting}
					bool c, t, float  p1, p2
					while(c == 1):
						if t:
							d1 = bernoulli.rvs(size=1, p=p1)[0]
							if d1:
								c = 0
							else:
								t = not t
						else:
							d2 = bernoulli.rvs(size=1, p=p2)[0]
							if d2:
								c = 0
							else:
								t = not t
				\end{lstlisting}
				\caption{\duel: Duelling cowboys}
				\label{fig:ex19}
			\end{figure*}

			\begin{figure*}
				\begin{lstlisting}
					int x, count
					while(x <= 10):
						x = x + 1
						count = count + 1
				\end{lstlisting}
				\caption{\detm: Deterministic loop}
				\label{fig:ex15a}
			\end{figure*}

			\begin{figure*}
				\begin{lstlisting}
					int x, z, float p
					while(x-1 >= 0):
						d = bernoulli.rvs(size=1, p=p)[0]
						if(d):
							x = x - 1
						z = z + 1
				\end{lstlisting}
				\caption{\revbin: A ``reversed'' binomial distribution}
				\label{fig:ex21}
			\end{figure*}

		\end{appendices}
	\clearpage
	\bibliography{fmsd}


\begin{thebibliography}{46}
\ifx \bisbn   \undefined \def \bisbn  #1{ISBN #1}\fi
\ifx \binits  \undefined \def \binits#1{#1}\fi
\ifx \bauthor  \undefined \def \bauthor#1{#1}\fi
\ifx \batitle  \undefined \def \batitle#1{#1}\fi
\ifx \bjtitle  \undefined \def \bjtitle#1{#1}\fi
\ifx \bvolume  \undefined \def \bvolume#1{\textbf{#1}}\fi
\ifx \byear  \undefined \def \byear#1{#1}\fi
\ifx \bissue  \undefined \def \bissue#1{#1}\fi
\ifx \bfpage  \undefined \def \bfpage#1{#1}\fi
\ifx \blpage  \undefined \def \blpage #1{#1}\fi
\ifx \burl  \undefined \def \burl#1{\textsf{#1}}\fi
\ifx \doiurl  \undefined \def \doiurl#1{\url{https://doi.org/#1}}\fi
\ifx \betal  \undefined \def \betal{\textit{et al.}}\fi
\ifx \binstitute  \undefined \def \binstitute#1{#1}\fi
\ifx \binstitutionaled  \undefined \def \binstitutionaled#1{#1}\fi
\ifx \bctitle  \undefined \def \bctitle#1{#1}\fi
\ifx \beditor  \undefined \def \beditor#1{#1}\fi
\ifx \bpublisher  \undefined \def \bpublisher#1{#1}\fi
\ifx \bbtitle  \undefined \def \bbtitle#1{#1}\fi
\ifx \bedition  \undefined \def \bedition#1{#1}\fi
\ifx \bseriesno  \undefined \def \bseriesno#1{#1}\fi
\ifx \blocation  \undefined \def \blocation#1{#1}\fi
\ifx \bsertitle  \undefined \def \bsertitle#1{#1}\fi
\ifx \bsnm \undefined \def \bsnm#1{#1}\fi
\ifx \bsuffix \undefined \def \bsuffix#1{#1}\fi
\ifx \bparticle \undefined \def \bparticle#1{#1}\fi
\ifx \barticle \undefined \def \barticle#1{#1}\fi
\bibcommenthead
\ifx \bconfdate \undefined \def \bconfdate #1{#1}\fi
\ifx \botherref \undefined \def \botherref #1{#1}\fi
\ifx \url \undefined \def \url#1{\textsf{#1}}\fi
\ifx \bchapter \undefined \def \bchapter#1{#1}\fi
\ifx \bbook \undefined \def \bbook#1{#1}\fi
\ifx \bcomment \undefined \def \bcomment#1{#1}\fi
\ifx \oauthor \undefined \def \oauthor#1{#1}\fi
\ifx \citeauthoryear \undefined \def \citeauthoryear#1{#1}\fi
\ifx \endbibitem  \undefined \def \endbibitem {}\fi
\ifx \bconflocation  \undefined \def \bconflocation#1{#1}\fi
\ifx \arxivurl  \undefined \def \arxivurl#1{\textsf{#1}}\fi
\csname PreBibitemsHook\endcsname

\bibitem[\protect\citeauthoryear{Kozen}{1981}]{Kozen81}
\begin{botherref}
\oauthor{\bsnm{Kozen}, \binits{D.}}:
Semantics of probabilistic programs
\textbf{22}(3)
(1981)
\doiurl{10.1016/0022-0000(81)90036-2}
\end{botherref}
\endbibitem

\bibitem[\protect\citeauthoryear{Smith et~al.}{2019}]{SHA18}
\begin{bchapter}
\bauthor{\bsnm{Smith}, \binits{C.}},
\bauthor{\bsnm{Hsu}, \binits{J.}},
\bauthor{\bsnm{Albarghouthi}, \binits{A.}}:
\bctitle{Trace abstraction modulo probability}.
In: \bbtitle{POPL}
(\byear{2019}).
\doiurl{10.1145/3290352}
\end{bchapter}
\endbibitem

\bibitem[\protect\citeauthoryear{Albarghouthi and Hsu}{2018}]{AH17}
\begin{bchapter}
\bauthor{\bsnm{Albarghouthi}, \binits{A.}},
\bauthor{\bsnm{Hsu}, \binits{J.}}:
\bctitle{Synthesizing coupling proofs of differential privacy}.
In: \bbtitle{POPL}
(\byear{2018}).
\doiurl{10.1145/3158146}
\end{bchapter}
\endbibitem

\bibitem[\protect\citeauthoryear{Carbin
  et~al.}{2013}]{DBLP:conf/oopsla/CarbinMR13}
\begin{bchapter}
\bauthor{\bsnm{Carbin}, \binits{M.}},
\bauthor{\bsnm{Misailovic}, \binits{S.}},
\bauthor{\bsnm{Rinard}, \binits{M.C.}}:
\bctitle{Verifying quantitative reliability for programs that execute on
  unreliable hardware}.
In: \bbtitle{OOPSLA}
(\byear{2013}).
\doiurl{10.1145/2509136.2509546}
\end{bchapter}
\endbibitem

\bibitem[\protect\citeauthoryear{Roy et~al.}{2021}]{Kolahal}
\begin{bchapter}
\bauthor{\bsnm{Roy}, \binits{S.}},
\bauthor{\bsnm{Hsu}, \binits{J.}},
\bauthor{\bsnm{Albarghouthi}, \binits{A.}}:
\bctitle{Learning differentially private mechanisms}.
In: \bbtitle{SP}
(\byear{2021}).
\doiurl{10.1109/SP40001.2021.00060}
\end{bchapter}
\endbibitem

\bibitem[\protect\citeauthoryear{Baier
  et~al.}{1997}]{DBLP:conf/icalp/BaierCHKR97}
\begin{bchapter}
\bauthor{\bsnm{Baier}, \binits{C.}},
\bauthor{\bsnm{Clarke}, \binits{E.M.}},
\bauthor{\bsnm{Hartonas{-}Garmhausen}, \binits{V.}},
\bauthor{\bsnm{Kwiatkowska}, \binits{M.Z.}},
\bauthor{\bsnm{Ryan}, \binits{M.}}:
\bctitle{Symbolic model checking for probabilistic processes}.
In: \bbtitle{ICALP}
(\byear{1997})
\end{bchapter}
\endbibitem

\bibitem[\protect\citeauthoryear{Kwiatkowska
  et~al.}{2011}]{kwiatkowska2011prism}
\begin{bchapter}
\bauthor{\bsnm{Kwiatkowska}, \binits{M.}},
\bauthor{\bsnm{Norman}, \binits{G.}},
\bauthor{\bsnm{Parker}, \binits{D.}}:
\bctitle{{PRISM} 4.0: Verification of probabilistic real-time systems}.
In: \bbtitle{CAV}
(\byear{2011}).
\doiurl{10.1007/978-3-642-22110-1\_47}
\end{bchapter}
\endbibitem

\bibitem[\protect\citeauthoryear{Dehnert
  et~al.}{2017}]{DBLP:journals/corr/DehnertJK017}
\begin{bchapter}
\bauthor{\bsnm{Dehnert}, \binits{C.}},
\bauthor{\bsnm{Junges}, \binits{S.}},
\bauthor{\bsnm{Katoen}, \binits{J.}},
\bauthor{\bsnm{Volk}, \binits{M.}}:
\bctitle{A storm is coming: {A} modern probabilistic model checker}.
In: \bbtitle{CAV}
(\byear{2017}).
\doiurl{10.1007/978-3-319-63390-9\_31}
\end{bchapter}
\endbibitem

\bibitem[\protect\citeauthoryear{Kozen}{1985}]{Kozen:1985}
\begin{botherref}
\oauthor{\bsnm{Kozen}, \binits{D.}}:
A probabilistic {PDL}
\textbf{30}(2)
(1985)
\doiurl{10.1016/0022-0000(85)90012-1}
\end{botherref}
\endbibitem

\bibitem[\protect\citeauthoryear{Morgan et~al.}{1996}]{Morgan:1996}
\begin{barticle}
\bauthor{\bsnm{Morgan}, \binits{C.}},
\bauthor{\bsnm{McIver}, \binits{A.}},
\bauthor{\bsnm{Seidel}, \binits{K.}}:
\batitle{Probabilistic predicate transformers}.
\bjtitle{TOPLAS}
(\byear{1996})
\doiurl{10.1145/229542.229547}
\end{barticle}
\endbibitem

\bibitem[\protect\citeauthoryear{{McIver} and Morgan}{2005}]{McIver:2005}
\begin{bbook}
\bauthor{\bsnm{{McIver}}, \binits{A.}},
\bauthor{\bsnm{Morgan}, \binits{C.}}:
\bbtitle{Abstraction, Refinement, and Proof for Probabilistic Systems},
(\byear{2005}).
\doiurl{10.1007/b138392}
\end{bbook}
\endbibitem

\bibitem[\protect\citeauthoryear{Dijkstra}{1975}]{dijkstra-wp}
\begin{bchapter}
\bauthor{\bsnm{Dijkstra}, \binits{E.W.}}:
\bctitle{Guarded commands, non-determinancy and a calculus for the derivation
  of programs}.
In: \bbtitle{Language Hierarchies and Interfaces}
(\byear{1975}).
\doiurl{10.1007/3-540-07994-7\_51}
\end{bchapter}
\endbibitem

\bibitem[\protect\citeauthoryear{Gretz et~al.}{2013}]{DBLP:conf/qest/GretzKM13}
\begin{bchapter}
\bauthor{\bsnm{Gretz}, \binits{F.}},
\bauthor{\bsnm{Katoen}, \binits{J.}},
\bauthor{\bsnm{McIver}, \binits{A.}}:
\bctitle{Prinsys - on a quest for probabilistic loop invariants}.
In: \bbtitle{QEST}
(\byear{2013}).
\doiurl{10.1007/978-3-642-40196-1\_17}
\end{bchapter}
\endbibitem

\bibitem[\protect\citeauthoryear{Chen et~al.}{2015}]{DBLP:conf/cav/ChenHWZ15}
\begin{bchapter}
\bauthor{\bsnm{Chen}, \binits{Y.}},
\bauthor{\bsnm{Hong}, \binits{C.}},
\bauthor{\bsnm{Wang}, \binits{B.}},
\bauthor{\bsnm{Zhang}, \binits{L.}}:
\bctitle{Counterexample-guided polynomial loop invariant generation by
  {Lagrange} interpolation}.
In: \bbtitle{CAV}
(\byear{2015}).
\doiurl{10.1007/978-3-319-21690-4\_44}
\end{bchapter}
\endbibitem

\bibitem[\protect\citeauthoryear{Flanagan and
  Leino}{2001}]{DBLP:conf/fm/FlanaganL01}
\begin{bchapter}
\bauthor{\bsnm{Flanagan}, \binits{C.}},
\bauthor{\bsnm{Leino}, \binits{K.R.M.}}:
\bctitle{Houdini, an annotation assistant for esc/java}.
In: \bbtitle{FME}
(\byear{2001}).
\doiurl{10.1007/3-540-45251-6\_29}
\end{bchapter}
\endbibitem

\bibitem[\protect\citeauthoryear{Ernst et~al.}{2007}]{daikon}
\begin{barticle}
\bauthor{\bsnm{Ernst}, \binits{M.D.}},
\bauthor{\bsnm{Perkins}, \binits{J.H.}},
\bauthor{\bsnm{Guo}, \binits{P.J.}},
\bauthor{\bsnm{McCamant}, \binits{S.}},
\bauthor{\bsnm{Pacheco}, \binits{C.}},
\bauthor{\bsnm{Tschantz}, \binits{M.S.}},
\bauthor{\bsnm{Xiao}, \binits{C.}}:
\batitle{The {Daikon} system for dynamic detection of likely invariants}.
\bjtitle{Sci. Comput. Program.}
(\byear{2007})
\doiurl{10.1016/j.scico.2007.01.015}
\end{barticle}
\endbibitem

\bibitem[\protect\citeauthoryear{Quinlan}{1992}]{quinlan1992learning}
\begin{bchapter}
\bauthor{\bsnm{Quinlan}, \binits{J.R.}}:
\bctitle{Learning with continuous classes}.
In: \bbtitle{AJCAI},
vol. \bseriesno{92}
(\byear{1992})
\end{bchapter}
\endbibitem

\bibitem[\protect\citeauthoryear{Yang
  et~al.}{2018}]{DBLP:journals/corr/abs-1806-06988}
\begin{botherref}
\oauthor{\bsnm{Yang}, \binits{Y.}},
\oauthor{\bsnm{Morillo}, \binits{I.G.}},
\oauthor{\bsnm{Hospedales}, \binits{T.M.}}:
Deep neural decision trees.
CoRR
(2018)
{\href{https://arxiv.org/abs/1806.06988}{{1806.06988}}}
\end{botherref}
\endbibitem

\bibitem[\protect\citeauthoryear{Chatterjee et~al.}{2016a}]{Chatterjee2016}
\begin{bchapter}
\bauthor{\bsnm{Chatterjee}, \binits{K.}},
\bauthor{\bsnm{Fu}, \binits{H.}},
\bauthor{\bsnm{Goharshady}, \binits{A.K.}}:
\bctitle{Termination analysis of probabilistic programs through
  {Positivstellensatz's}}.
In: \bbtitle{CAV}
(\byear{2016})
\end{bchapter}
\endbibitem

\bibitem[\protect\citeauthoryear{Chatterjee
  et~al.}{2016b}]{Chatterjee:2016:AAQ:2837614.2837639}
\begin{bchapter}
\bauthor{\bsnm{Chatterjee}, \binits{K.}},
\bauthor{\bsnm{Fu}, \binits{H.}},
\bauthor{\bsnm{Novotn\'{y}}, \binits{P.}},
\bauthor{\bsnm{Hasheminezhad}, \binits{R.}}:
\bctitle{Algorithmic analysis of qualitative and quantitative termination
  problems for affine probabilistic programs}.
In: \bbtitle{POPL}
(\byear{2016}).
\doiurl{10.1145/2837614.2837639}
\end{bchapter}
\endbibitem

\bibitem[\protect\citeauthoryear{McIver et~al.}{2018}]{mciver2016new}
\begin{bchapter}
\bauthor{\bsnm{McIver}, \binits{A.}},
\bauthor{\bsnm{Morgan}, \binits{C.}},
\bauthor{\bsnm{Kaminski}, \binits{B.L.}},
\bauthor{\bsnm{Katoen}, \binits{J.}}:
\bctitle{A new proof rule for almost-sure termination}.
In: \bbtitle{POPL}
(\byear{2018}).
\doiurl{10.1145/3158121}
\end{bchapter}
\endbibitem

\bibitem[\protect\citeauthoryear{Batz
  et~al.}{2021}]{DBLP:journals/pacmpl/BatzKKM21}
\begin{bchapter}
\bauthor{\bsnm{Batz}, \binits{K.}},
\bauthor{\bsnm{Kaminski}, \binits{B.L.}},
\bauthor{\bsnm{Katoen}, \binits{J.}},
\bauthor{\bsnm{Matheja}, \binits{C.}}:
\bctitle{Relatively complete verification of probabilistic programs: an
  expressive language for expectation-based reasoning}.
In: \bbtitle{{POPL}}
(\byear{2021}).
\doiurl{10.1145/3434320}
\end{bchapter}
\endbibitem

\bibitem[\protect\citeauthoryear{Kaminski et~al.}{2016}]{KaminskiKMO16}
\begin{bchapter}
\bauthor{\bsnm{Kaminski}, \binits{B.L.}},
\bauthor{\bsnm{Katoen}, \binits{J.}},
\bauthor{\bsnm{Matheja}, \binits{C.}},
\bauthor{\bsnm{Olmedo}, \binits{F.}}:
\bctitle{Weakest precondition reasoning for expected run-times of probabilistic
  programs}.
In: \bbtitle{ESOP}
(\byear{2016}).
\doiurl{10.1007/978-3-662-49498-1\_15}
\end{bchapter}
\endbibitem

\bibitem[\protect\citeauthoryear{Hark et~al.}{2020}]{hark2019aiming}
\begin{bchapter}
\bauthor{\bsnm{Hark}, \binits{M.}},
\bauthor{\bsnm{Kaminski}, \binits{B.L.}},
\bauthor{\bsnm{Giesl}, \binits{J.}},
\bauthor{\bsnm{Katoen}, \binits{J.}}:
\bctitle{Aiming low is harder: induction for lower bounds in probabilistic
  program verification}.
(\byear{2020}).
\doiurl{10.1145/3371105}
\end{bchapter}
\endbibitem

\bibitem[\protect\citeauthoryear{Park}{1969}]{park1969fixpoint}
\begin{botherref}
\oauthor{\bsnm{Park}, \binits{D.}}:
Fixpoint induction and proofs of program properties.
Machine intelligence
\textbf{5}
(1969)
\end{botherref}
\endbibitem

\bibitem[\protect\citeauthoryear{Kaminski and
  Katoen}{2017}]{kaminski2017weakest}
\begin{bchapter}
\bauthor{\bsnm{Kaminski}, \binits{B.L.}},
\bauthor{\bsnm{Katoen}, \binits{J.-P.}}:
\bctitle{A weakest pre-expectation semantics for mixed-sign expectations}.
In: \bbtitle{LICS}
(\byear{2017}).
\doiurl{10.5555/3329995.3330088}
\end{bchapter}
\endbibitem

\bibitem[\protect\citeauthoryear{Leighton and Rubinfeld}{2006}]{gamblernotes}
\begin{botherref}
\oauthor{\bsnm{Leighton}, \binits{T.}},
\oauthor{\bsnm{Rubinfeld}, \binits{R.}}:
Random Walks -- Lecture notes in Mathematics for Computer Science.
MIT CS 6.042/18.062
(2006).
\url{https://web.mit.edu/neboat/Public/6.042/randomwalks.pdf}
\end{botherref}
\endbibitem

\bibitem[\protect\citeauthoryear{Huang et~al.}{2019}]{huang2019modular}
\begin{barticle}
\bauthor{\bsnm{Huang}, \binits{M.}},
\bauthor{\bsnm{Fu}, \binits{H.}},
\bauthor{\bsnm{Chatterjee}, \binits{K.}},
\bauthor{\bsnm{Goharshady}, \binits{A.K.}}:
\batitle{Modular verification for almost-sure termination of probabilistic
  programs}.
\bjtitle{Proceedings of the ACM on Programming Languages}
\bvolume{3}(\bissue{OOPSLA}),
\bfpage{1}--\blpage{29}
(\byear{2019})
\end{barticle}
\endbibitem

\bibitem[\protect\citeauthoryear{Majumdar and
  Sathiyanarayana}{2024}]{majumdar2024sound}
\begin{botherref}
\oauthor{\bsnm{Majumdar}, \binits{R.}},
\oauthor{\bsnm{Sathiyanarayana}, \binits{V.}}:
Sound and complete proof rules for probabilistic termination.
arXiv preprint arXiv:2404.19724
(2024)
\end{botherref}
\endbibitem

\bibitem[\protect\citeauthoryear{Bartocci et~al.}{2020}]{bartocci2020mora}
\begin{bchapter}
\bauthor{\bsnm{Bartocci}, \binits{E.}},
\bauthor{\bsnm{Kov{\'a}cs}, \binits{L.}},
\bauthor{\bsnm{Stankovi{\v{c}}}, \binits{M.}}:
\bctitle{Mora-automatic generation of moment-based invariants}.
In: \bbtitle{TACAS}
(\byear{2020}).
\doiurl{10.1007/978-3-030-45190-5\_28}
\end{bchapter}
\endbibitem

\bibitem[\protect\citeauthoryear{Feng et~al.}{2017}]{feng2017finding}
\begin{bchapter}
\bauthor{\bsnm{Feng}, \binits{Y.}},
\bauthor{\bsnm{Zhang}, \binits{L.}},
\bauthor{\bsnm{Jansen}, \binits{D.N.}},
\bauthor{\bsnm{Zhan}, \binits{N.}},
\bauthor{\bsnm{Xia}, \binits{B.}}:
\bctitle{Finding polynomial loop invariants for probabilistic programs}.
In: \bbtitle{ATVA}
(\byear{2017})
\end{bchapter}
\endbibitem

\bibitem[\protect\citeauthoryear{Batz
  et~al.}{2023}]{DBLP:conf/tacas/BatzCJKKM23}
\begin{bchapter}
\bauthor{\bsnm{Batz}, \binits{K.}},
\bauthor{\bsnm{Chen}, \binits{M.}},
\bauthor{\bsnm{Junges}, \binits{S.}},
\bauthor{\bsnm{Kaminski}, \binits{B.L.}},
\bauthor{\bsnm{Katoen}, \binits{J.-P.}},
\bauthor{\bsnm{Matheja}, \binits{C.}}:
\bctitle{Probabilistic program verification via inductive synthesis of
  inductive invariants}.
In: \bbtitle{International Conference on Tools and Algorithms for the
  Construction and Analysis of Systems},
pp. \bfpage{410}--\blpage{429}
(\byear{2023}).
\bcomment{Springer}
\end{bchapter}
\endbibitem

\bibitem[\protect\citeauthoryear{Chakarov and
  Sankaranarayanan}{2014}]{DBLP:conf/sas/ChakarovS14}
\begin{bchapter}
\bauthor{\bsnm{Chakarov}, \binits{A.}},
\bauthor{\bsnm{Sankaranarayanan}, \binits{S.}}:
\bctitle{Expectation invariants for probabilistic program loops as fixed
  points}.
In: \bbtitle{SAS}
(\byear{2014}).
\doiurl{10.1007/978-3-319-10936-7\_6}
\end{bchapter}
\endbibitem

\bibitem[\protect\citeauthoryear{Wang et~al.}{2018}]{DBLP:conf/pldi/WangHR18}
\begin{bchapter}
\bauthor{\bsnm{Wang}, \binits{D.}},
\bauthor{\bsnm{Hoffmann}, \binits{J.}},
\bauthor{\bsnm{Reps}, \binits{T.W.}}:
\bctitle{{PMAF:} an algebraic framework for static analysis of probabilistic
  programs}.
In: \bbtitle{PLDI}
(\byear{2018}).
\doiurl{10.1145/3192366.3192408}
\end{bchapter}
\endbibitem

\bibitem[\protect\citeauthoryear{Chakarov and
  Sankaranarayanan}{2013}]{Chakarov-martingale}
\begin{bchapter}
\bauthor{\bsnm{Chakarov}, \binits{A.}},
\bauthor{\bsnm{Sankaranarayanan}, \binits{S.}}:
\bctitle{Probabilistic program analysis with martingales}.
In: \bbtitle{CAV}
(\byear{2013}).
\doiurl{10.1007/978-3-642-39799-8\_34}
\end{bchapter}
\endbibitem

\bibitem[\protect\citeauthoryear{Barthe et~al.}{2016}]{BEFFH16}
\begin{bchapter}
\bauthor{\bsnm{Barthe}, \binits{G.}},
\bauthor{\bsnm{Espitau}, \binits{T.}},
\bauthor{\bsnm{Ferrer~Fioriti}, \binits{L.M.}},
\bauthor{\bsnm{Hsu}, \binits{J.}}:
\bctitle{Synthesizing probabilistic invariants via {Doob}'s decomposition}.
In: \bbtitle{CAV}
(\byear{2016}).
\doiurl{10.1007/978-3-319-41528-4\_3}
\end{bchapter}
\endbibitem

\bibitem[\protect\citeauthoryear{Bartocci et~al.}{2019}]{bartocci2019automatic}
\begin{bchapter}
\bauthor{\bsnm{Bartocci}, \binits{E.}},
\bauthor{\bsnm{Kov{\'a}cs}, \binits{L.}},
\bauthor{\bsnm{Stankovi{\v{c}}}, \binits{M.}}:
\bctitle{Automatic generation of moment-based invariants for prob-solvable
  loops}.
In: \bbtitle{ATVA}
(\byear{2019}).
\doiurl{10.1007/978-3-030-31784-3\_15}
\end{bchapter}
\endbibitem

\bibitem[\protect\citeauthoryear{Kura et~al.}{2019}]{kura2019tail}
\begin{bchapter}
\bauthor{\bsnm{Kura}, \binits{S.}},
\bauthor{\bsnm{Urabe}, \binits{N.}},
\bauthor{\bsnm{Hasuo}, \binits{I.}}:
\bctitle{Tail probabilities for randomized program runtimes via martingales for
  higher moments}.
In: \bbtitle{TACAS}
(\byear{2019}).
\doiurl{10.1007/978-3-030-17465-1\_8}
\end{bchapter}
\endbibitem

\bibitem[\protect\citeauthoryear{Wang et~al.}{2021}]{wang2021central}
\begin{bchapter}
\bauthor{\bsnm{Wang}, \binits{D.}},
\bauthor{\bsnm{Hoffmann}, \binits{J.}},
\bauthor{\bsnm{Reps}, \binits{T.}}:
\bctitle{Central moment analysis for cost accumulators in probabilistic
  programs}.
In: \bbtitle{PLDI}
(\byear{2021}).
\doiurl{10.1145/3453483.3454062}
\end{bchapter}
\endbibitem

\bibitem[\protect\citeauthoryear{Abate et~al.}{2021}]{DBLP:conf/cav/AbateGR20}
\begin{bchapter}
\bauthor{\bsnm{Abate}, \binits{A.}},
\bauthor{\bsnm{Giacobbe}, \binits{M.}},
\bauthor{\bsnm{Roy}, \binits{D.}}:
\bctitle{Learning probabilistic termination proofs}.
In: \bbtitle{CAV}
(\byear{2021}).
\doiurl{10.1007/978-3-030-81688-9\_1}
\end{bchapter}
\endbibitem

\bibitem[\protect\citeauthoryear{Garg et~al.}{2016}]{icedt}
\begin{bchapter}
\bauthor{\bsnm{Garg}, \binits{P.}},
\bauthor{\bsnm{Neider}, \binits{D.}},
\bauthor{\bsnm{Madhusudan}, \binits{P.}},
\bauthor{\bsnm{Roth}, \binits{D.}}:
\bctitle{Learning invariants using decision trees and implication
  counterexamples}.
In: \bbtitle{POPL}
(\byear{2016}).
\doiurl{10.1145/2914770.2837664}
\end{bchapter}
\endbibitem

\bibitem[\protect\citeauthoryear{Miltner et~al.}{2020}]{Miltner2020}
\begin{bchapter}
\bauthor{\bsnm{Miltner}, \binits{A.}},
\bauthor{\bsnm{Padhi}, \binits{S.}},
\bauthor{\bsnm{Millstein}, \binits{T.}},
\bauthor{\bsnm{Walker}, \binits{D.}}:
\bctitle{Data-driven inference of representation invariants}.
In: \bbtitle{PLDI 20}
(\byear{2020}).
\doiurl{10.1145/3385412.3385967}
\end{bchapter}
\endbibitem

\bibitem[\protect\citeauthoryear{Solar-Lezama}{2013}]{sketch}
\begin{barticle}
\bauthor{\bsnm{Solar-Lezama}, \binits{A.}}:
\batitle{Program sketching}.
\bjtitle{Int. J. Softw. Tools Technol. Transf.}
(\byear{2013})
\doiurl{10.1007/s10009-012-0249-7}
\end{barticle}
\endbibitem

\bibitem[\protect\citeauthoryear{Si et~al.}{2018}]{code2inv}
\begin{bchapter}
\bauthor{\bsnm{Si}, \binits{X.}},
\bauthor{\bsnm{Dai}, \binits{H.}},
\bauthor{\bsnm{Raghothaman}, \binits{M.}},
\bauthor{\bsnm{Naik}, \binits{M.}},
\bauthor{\bsnm{Song}, \binits{L.}}:
\bctitle{Learning loop invariants for program verification}.
In: \bbtitle{NeurIPS}
(\byear{2018}).
\doiurl{10.5555/3327757.3327873}
\end{bchapter}
\endbibitem

\bibitem[\protect\citeauthoryear{Lahiri and Roy}{2022}]{Lahiri:22}
\begin{bchapter}
\bauthor{\bsnm{Lahiri}, \binits{S.}},
\bauthor{\bsnm{Roy}, \binits{S.}}:
\bctitle{Almost correct invariants: Synthesizing inductive invariants by
  fuzzing proofs}.
In: \bbtitle{ISSTA}
(\byear{2022})
\end{bchapter}
\endbibitem

\bibitem[\protect\citeauthoryear{Aguirre et~al.}{2021}]{ABHKKM19}
\begin{bchapter}
\bauthor{\bsnm{Aguirre}, \binits{A.}},
\bauthor{\bsnm{Barthe}, \binits{G.}},
\bauthor{\bsnm{Hsu}, \binits{J.}},
\bauthor{\bsnm{Kaminski}, \binits{B.L.}},
\bauthor{\bsnm{Katoen}, \binits{J.-P.}},
\bauthor{\bsnm{Matheja}, \binits{C.}}:
\bctitle{A pre-expectation calculus for probabilistic sensitivity}.
In: \bbtitle{POPL}
(\byear{2021}).
\doiurl{10.1145/3434333}
\end{bchapter}
\endbibitem

\end{thebibliography}
	\end{document}